%% file: main.tex
\documentclass[11pt]{article}
\usepackage[utf8]{inputenc} 
\usepackage[T1]{fontenc}
\usepackage{lmodern}

\usepackage{geometry}
\usepackage{microtype}  
\usepackage{array}
\usepackage{multirow}
\usepackage{amsmath} 
\usepackage{amssymb}
\usepackage{amsthm} 
\usepackage{thmtools}

\usepackage[procnumbered,ruled,vlined,linesnumbered]{algorithm2e}
 
\SetCommentSty{mycommfont} 

\usepackage{xcolor}  
\usepackage{xspace}
\usepackage{enumitem}    
\usepackage{tikz}   

\usepackage{comment} 

\usepackage{graphicx} 
\usepackage{caption} 
\usepackage{subcaption}  

\usepackage{authblk} 

\def\AdvCite{True} 

\ifdefined\AdvCite

\usepackage[
style=alphabetic, 
backref=true,
doi=false,
url=false,
maxcitenames=3,
mincitenames=3,
maxbibnames=10,
minbibnames=10,
backend=bibtex8,
sortlocale=en_US
]{biblatex}

\usepackage[colorlinks]{hyperref} 
\usepackage[capitalise]{cleveref}

\DeclareCiteCommand{\citem}
{}
{\mkbibbrackets{\bibhyperref{\usebibmacro{postnote}}}}
{\multicitedelim}
{}

\renewcommand*{\multicitedelim}{\addcomma\space}

\AtEveryCitekey{\clearfield{note}}

\newcommand{\myhref}[1]{%
	\iffieldundef{doi}
	{\iffieldundef{url}
		{#1}
		{\href{\strfield{url}}{#1}}}
	{\href{http://dx.doi.org/\strfield{doi}}{#1}}%
}

\DeclareFieldFormat{title}{\myhref{\mkbibemph{#1}}}
\DeclareFieldFormat
[article,inbook,incollection,inproceedings,patent,thesis,unpublished]
{title}{\myhref{\mkbibquote{#1\isdot}}}

\makeatletter
\AtBeginDocument{%
	\newlength{\temp@x}%
	\newlength{\temp@y}%
	\newlength{\temp@w}%
	\newlength{\temp@h}%
	\def\my@coords#1#2#3#4{%
		\setlength{\temp@x}{#1}%
		\setlength{\temp@y}{#2}%
		\setlength{\temp@w}{#3}%
		\setlength{\temp@h}{#4}%
		\adjustlengths{}%
		\my@pdfliteral{\strip@pt\temp@x\space\strip@pt\temp@y\space\strip@pt\temp@w\space\strip@pt\temp@h\space re}}%
	\ifpdf
	\typeout{In PDF mode}%
	\def\my@pdfliteral#1{\pdfliteral page{#1}}
	\def\adjustlengths{}%
	\fi
	\ifxetex
	\def\my@pdfliteral #1{}
	\def\adjustlengths{\setlength{\temp@h}{-\temp@h}\addtolength{\temp@y}{1in}\addtolength{\temp@x}{-1in}}%
	\fi%
	\def\Hy@colorlink#1{%
		\begingroup
		\ifHy@ocgcolorlinks
		\def\Hy@ocgcolor{#1}%
		\my@pdfliteral{q}%
		\my@pdfliteral{7 Tr}
		\else
		\HyColor@UseColor#1%
		\fi
	}%
	\def\Hy@endcolorlink{%
		\ifHy@ocgcolorlinks%
		\my@pdfliteral{/OC/OCPrint BDC}%
		\my@coords{0pt}{0pt}{\pdfpagewidth}{\pdfpageheight}%
		\my@pdfliteral{F}
		\my@pdfliteral{EMC/OC/OCView BDC}%
		\begingroup%
		\expandafter\HyColor@UseColor\Hy@ocgcolor%
		\my@coords{0pt}{0pt}{\pdfpagewidth}{\pdfpageheight}%
		\my@pdfliteral{F}
		\endgroup%
		\my@pdfliteral{EMC}%
		\my@pdfliteral{0 Tr}
		\my@pdfliteral{Q}%
		\fi
		\endgroup
	}%
} 
\makeatother

\else

\usepackage[ocgcolorlinks]{hyperref} 
\usepackage{cleveref}

\fi 

\input{commands}

\renewcommand{\paragraph}[1]{\medskip\noindent{\bf #1}\xspace}

\colorlet{DarkRed}{red!50!black} 
\colorlet{DarkGreen}{green!50!black}
\colorlet{DarkBlue}{blue!50!black}

\hypersetup{
	linkcolor = DarkRed,
	citecolor = DarkGreen,
	urlcolor = DarkBlue,
	bookmarks = true,
	bookmarksnumbered = true,
	linktocpage = true 
}

\declaretheorem[numberwithin=section]{theorem}
\declaretheorem[numberlike=theorem]{lemma}
\declaretheorem[numberlike=theorem]{proposition}
\declaretheorem[numberlike=theorem]{corollary}
\declaretheorem[numberlike=theorem]{claim}

\crefname{algorithm}{Algorithm}{Algorithms}
\Crefname{algorithm}{Algorithm}{Algorithms}

\theoremstyle{definition}

\DontPrintSemicolon
\SetKw{KwAnd}{and}
\SetProcNameSty{textsc}
\SetFuncSty{textsc}

\newcommand{\ot}{\tilde{O}}

\newcommand{\ignore}[1]{}

\newcommand{\cC}{\mathcal{C}\xspace}

\newcommand{\eps}{\epsilon}

\newcommand{\poly}{\operatorname{poly}} 
\newcommand{\polylog}{\operatorname{polylog}}

\newcommand{\vol}{\operatorname{vol}}

\newcommand{\textlocal}{\operatorname{local}} 
\newcommand{\textpair}{\operatorname{pair}} 
\newcommand{\textout}{\operatorname{out}} 
\newcommand{\logk}{\lfloor \log_2k \rfloor}   
\newcommand{\logeta}{\lfloor \log_2\eta_i\rfloor}

\newcommand{\gap}{\mathtt{gap}}%
 
 \newcommand{\textbig}{\operatorname{big}} 
 \newcommand{\textsmall}{\operatorname{small}} 
\newcommand{\out}{\operatorname{out}}  
\newcommand{\inn}{\operatorname{in}}  
\newcommand{\textin}{\operatorname{in}}

\newcommand{\earlybotline}{$7 $}   

\global\long\def\localEC{\text{LocalEC}}%
\global\long\def\DFS{\mathtt{DFS}}%

\def\ShowComment{True} 

\ifdefined\ShowComment
\def\thatchaphol#1{\marginpar{$\leftarrow$\fbox{T}}\footnote{$\Rightarrow$~{\sf\textcolor{purple}{#1 --Thatchaphol}}}}
\def\danupon#1{\marginpar{$\leftarrow$\fbox{D}}\footnote{$\Rightarrow$~{\sf\textcolor{orange}{#1 --Danupon}}}}
\def\sorrachai#1{\marginpar{$\leftarrow$\fbox{S}}\footnote{$\Rightarrow$~{\sf\textcolor{red}{#1 --Shen}}}}
\def\note#1{#1}

\else
\def\thatchaphol#1{}
\def\danupon#1{}
\def\shen#1{}
\def\sorrachai#1{}
\def\note#1{} 

\fi

\title{Computing and Testing Small Vertex Connectivity in Near-Linear Time and Queries}

\author[1]{Danupon Nanongkai}
\author[2]{Thatchaphol Saranurak\thanks{Work partially done while at KTH Royal Institute of Technology, Sweden.}}
\author[3]{Sorrachai Yingchareonthawornchai\thanks{Work partially done while at Michigan State University, USA.}}
\affil[1]{KTH Royal Institute of Technology, Sweden}
\affil[2]{Toyota Technological Institute at Chicago, USA}
\affil[3]{Aalto University, Finland}

\date{} 
\begin{document}

	\begin{titlepage}   
		\maketitle
		\pagenumbering{roman}   
		\input{abstract}  
		\newpage     
		\setcounter{tocdepth}{2}  
		\tableofcontents    
	\end{titlepage}

	\newpage         
	\pagenumbering{arabic}           
	
	\input{intro}
\input{prelim}

\input{ec_new}

\input{vc_new}
	\input{property_testing}

	\section*{Acknowledgement}  
        
        This project has received funding from the European Research
        Council (ERC) under the European Union's Horizon 2020 research
        and innovation programme under grant agreement No
        715672 and 759557. Nanongkai was also partially supported by the Swedish
        Research Council (Reg. No. 2015-04659.).

	\ifdefined\AdvCite

	\printbibliography[heading=bibintoc] 
	
	\else 
 	
	\bibliographystyle{alpha}
	\bibliography{references} 
	
	\fi

	 \appendix  
        \input{local}

\end{document}

%% file: commands.tex
\allowdisplaybreaks[1]

\makeatletter
\g@addto@macro\bfseries{\boldmath}
\makeatother

\makeatletter
\g@addto@macro\mdseries{\unboldmath}
\g@addto@macro\normalfont{\unboldmath}
\g@addto@macro\rmfamily{\unboldmath}
\g@addto@macro\upshape{\unboldmath}
\makeatother

%% file: abstract.tex

\begin{abstract}
We present a new, simple, algorithm for the \textit{local vertex connectivity} problem
(localVC) introduced by Nanongkai~et~al. [STOC'19]. Roughly, given an undirected unweighted graph $G$, a seed
vertex $x$, a target volume $\nu$, and a target separator size $k$, the goal of LocalVC is to 
remove $k$ vertices ``near'' $x$ (in terms of $\nu$)  to disconnect the graph in ``local time'', which depends only on parameters $\nu$ and $k$. Nanongkai~et~al. presented an $O(\nu^{1.5}k\polylog(\nu k))$-time deterministic algorithm for this problem.
In this paper, we present a simple randomized algorithm with running time $O(\nu k^2)$ and correctness probability $2/3$. Our algorithm is faster than the previous one when $k=O(\sqrt{\nu})$. We also can handle directed graphs and achieve $(1+\epsilon)$-approximation with even faster running time. 
	
Plugging our new localVC algorithm in the generic framework of Nanongkai~et~al. immediately lead to a randomized $\tilde O(m+nk^3)$-time algorithm for the classic $k$-vertex connectivity problem on undirected graphs. ($\tilde O(T)$ hides $\polylog(T)$.)
This is the {\em first} near-linear time  algorithm for any $4\leq k \leq \polylog n$. Previously, linear-time algorithms were known only for $k\leq 3$ [Tarjan FOCS’71; Hopcroft, Tarjan SICOMP’73], despite a linear-time algorithm being  postulated since 1974 in the book of Aho, Hopcroft and Ullman. Previous fastest algorithm for small $k$ takes $\tilde O(m+n^{4/3}k^{7/3})$ time [Nanongkai~et~al., STOC'19].  

This work is inspired by the algorithm of Chechik~et~al. [SODA'17] for computing the maximal $k$-edge connected subgraphs. In turn, our algorithms lead to some improvements over the bounds of Chechik~et~al. 
Forster and Yang [arXiv'19] has independently developed local algorithms similar to ours, and showed that they  
lead to an $\tilde O(k^3/\epsilon)$ bound for testing $k$-edge and -vertex connectivity, resolving two long-standing open problems in property testing since the work of Goldreich and Ron [STOC'97] and Orenstein and Ron [Theor. Comput. Sci.'11]. Inspired by this, we use local approximation algorithms to obtain bounds that are near-linear in $k$, namely $\tilde O(k/\epsilon)$ and $\tilde O(k/\epsilon^2)$ for the bounded and unbounded degree cases, respectively. 
For testing $k$-edge connectivity for simple graphs, the bound can be improved to $\tilde O(\min(k/\epsilon, 1/\epsilon^2))$. 
\end{abstract}
\begin{footnotesize}
	\noindent {\bf Independent work:} Independently from our result, Forster and Yang [2019] present local algorithms similar to ours and observed faster algorithms for computing the vertex connectivity and  the maximal $k$-edge connected subgraphs (with lower running time than ours for the second problem).\footnote{\label{foot:independent_result}On April 17, 2019, our result was announced in the TCS+ talk by Thatchaphol Saranurak (\url{https://youtu.be/V1kq1filhjk}) and Forster and Yang announced their result at \url{https://arxiv.org/abs/1904.08382}.}
	They additionally observed that this leads to resolving two open problems in property testing, but did not observe an application to approximating the vertex connectivity. Our bounds for  testing connectivities were inspired by their observation.  
\end{footnotesize}

%% file: intro.tex
\section{Introduction}\label{sec:intro}

Vertex connectivity is a basic graph-theoretic concept. It concerns the smallest {\em vertex cut} where a 
set  $S$ of vertices is a vertex cut of a graph $G$ if its removal disconnects some vertex $u\notin S$ 
from another vertex $v\notin S$.   (When $G$ is directed, this means that there is no directed path from  $u$ to $v$ in the remaining graph.) 
The {\em vertex connectivity} of $G$, denoted by $\kappa_{G}$ , is the size of the smallest vertex cut. The goal of the vertex connectivity problem is to compute $\kappa_G$ and  the smallest vertex cut. 
In this paper, we present a new, simple, algorithm for the {\em local} version of this problem, leading to almost optimal bounds for computing and approximating $\kappa_G$ when $\kappa_G$ is small, and other applications. 
For simplicity, our discussions below focus on exact algorithms for undirected graphs. 

%
%
%

\paragraph{Local Vertex Connectivity (LocalVC).} This problem concerns finding a vertex cut ``near'' a given vertex $x$. More precisely, for any vertex $v\in V$, define $N(v)$ to the set of neighbors of $v$, $\deg(v)=|N(v)|$, $N(L)=(\bigcup_{v\in L} N(v))\setminus L$, and $\vol(L)=\sum_{v\in L} \deg(v)$ (we call $\vol(L)$ the {\em volume} of $L$).  
Given a vertex $x$ and two integers $\nu$ and $k$, the LocalVC problem concerns the set $L\subseteq V$ such that
\begin{align}\label{eq:intro:LocalVC}
x\in L, \ \mbox{$N(L)$ is a vertex cut of size less than $k$},\ \mbox{and}\ \vol(L)\leq \nu\,.
\end{align} 
In other words, we are interested in a small vertex cut $N(L)$  that is ``near'' $x$ in the sense that $L$ has small volume.
An algorithm for this problem takes as input  $x$, $k$, $\nu$, and a pointer to an adjacency-list representation of $G$, and either
\begin{itemize}[noitemsep]
	\item outputs that  $L\subseteq V$ satisfying \Cref{eq:intro:LocalVC} does not exist, or 
	\item returns a vertex cut $S$ of size less than $k$.
\end{itemize}

Nanongkai~et~al.~\cite{NanongkaiSY19} recently introduced the LocalVC problem and designed a deterministic algorithm that takes $O(\nu^{1.5}k\polylog(\nu k))$ time under mild conditions. 
In this paper, we present a simple  randomized (Monte Carlo)  algorithm that takes $O(\nu k^2)$ time under the same conditions. 
\begin{theorem}[Main Result]
	\label{thm:intro:LocalVC_approx}There is a randomized (Monte Carlo) algorithm
	that takes as input a vertex $x\in V$ of an $n$-vertex $m$-edge
	graph $G=(V,E)$ represented
	as adjacency lists, and integers $k < n/4$ and $\nu<m/(8320k)$ 
	and runs in $O(\nu k^2)$ time to output either
	\begin{itemize}[noitemsep]
		\item the ``$\bot$'' symbol indicating that, with probability at least $1/2$,  $L\subseteq V$ satisfying \Cref{eq:intro:LocalVC} does not exist,  or
		\item a vertex cut $S$ of size less than $k$.
	\end{itemize}
\end{theorem}

%
%

%
%
%
%
%
%
%
%
%
%
%
%
%
%
%
%
%
%
%
%
%
%
%
%
%
%
%
%

%

%
%
%
%
Note that the error probability $1/2$ above can be made arbitrarily small by repeating the algorithm. 
Compared to the previous algorithm of Nanongkai~et~al., our algorithm is faster when $k\leq \sqrt{\nu}$.
It is worth noting that one can also derive an $(\nu k^{O(k)})$-time algorithm from the techniques of Chechik~et~al. \cite{ChechikHILP17} and some slower algorithms in the context of property testing (e.g. \cite{GoldreichR02,OrensteinR11,YoshidaI12,YoshidaI10}). 
Our algorithm is in fact very simple: it simply repeatedly finds a path starting at $x$ and ending at some random vertex. Our analysis is also very simple.

\paragraph{(Global) Vertex Connectivity.} 
The main application of our result is efficient algorithms for the vertex connectivity problem. 
There has been a long line
of research on  this problem since at least five decades ago (e.g. \cite{Kleitman1969methods,Podderyugin1973algorithm,EvenT75,Even75,Galil80,EsfahanianH84,Matula87,BeckerDDHKKMNRW82,LinialLW88,CheriyanT91,NagamochiI92,CheriyanR94,Henzinger97,HenzingerRG00,Gabow06,Censor-HillelGK14}).  (See Nanongkai~et~al.~\cite{NanongkaiSY19} for a more comprehensive literature survey.)
For the undirected case,  Aho,
Hopcroft and Ullman \cite[Problem 5.30]{AhoHU74} asked in their 1974 book for an
$O(m)$-time algorithm for
computing $\kappa_{G}$. 
Prior to our result, $O(m)$-time algorithms were known only when $\kappa_G \leq 3$, due to the classic results of Tarjan~\cite{Tarjan72} and Hopcroft and Tarjan~\cite{HopcroftT73}. In this paper, we present an algorithm that takes near-linear time whenever  $\kappa_G=O(\polylog(n))$. 
In this paper, we obtain the first algorithm in many decades that guarantees a near-linear time complexity for higher values of $\kappa_G$.

\begin{theorem}
	\label{thm:intro:VC_undir}There is a randomized algorithm
	that takes as input an undirected graph $G$ and,  with high probability, in time $\tilde{O}(m+nk^{3})$
	outputs a vertex cut $S$ of size  $k=\kappa_G$.\footnote{As usual, with high probability (w.h.p.) means with probability at least $1-1/n^c$ for arbitrary constant $c\geq 1$.}
\end{theorem}

The above result is near-linear time whenever $k=O(\polylog(n))$. 
By combining with previous results 
(e.g.  \cite{HenzingerRG00,LinialLW88}),
the best running time for solving vertex connectivity is $\tilde{O}(m+\min\{nk^{3}, n^2k, n^{\omega} + nk^{\omega} \})$. Prior to our work, the best running time for $k>3$ was  $\tilde{O}(m+  \min\{n^{4/3}k^{7/3}, n^2k, n^{\omega} + nk^{\omega} \} )$ \cite{NanongkaiSY19, HenzingerRG00, LinialLW88}.  In particular, we have an improved running time when $k\leq O(n^{0.457})$.

This result is obtained essentially by plugging in our LocalVC algorithm to the recent framework of Nanongkai~et~al. \cite{NanongkaiSY19}. The overall algorithm is fairly simple: Let $L$ be such that $N(L)$ is the optimal vertex cut. We guess the values $\nu=\vol(L)$ and $k=\kappa_G$, and run our LocalVC algorithm with parameters $\nu$ and $k$ on $n/\nu$ randomly-selected seed nodes $x$. 

\paragraph{Approximation Algorithms and Directed Graphs.} Results in \Cref{thm:intro:LocalVC_approx,thm:intro:VC_undir} can be generalized to $(1+\epsilon)$-approximation algorithms and to algorithms on directed graphs. The approximation guarantee means that the output vertex cut $S$ is of size less than $\left\lfloor (1+\epsilon)k\right\rfloor$.
The time complexity for LocalVC is $O(\nu k/\epsilon)$. This improves the $\tilde O(\nu^{1.5}/(\sqrt{k}\epsilon^{1.5}))$-time algorithm of Nanongkai et~al. \cite{NanongkaiSY19} when $k\leq (\nu/\epsilon)^{1/3}$.
For approximating $\kappa_G$, the time complexity is $\tilde{O}(\min\{mk/\epsilon , n^{2+o(1)}\sqrt{k}/\poly(\epsilon)\})$ where $k=\kappa_G$. 

Observe that the time complexities for exact algorithms in \Cref{thm:intro:LocalVC_approx,thm:intro:VC_undir} can be obtained by setting $\epsilon=1/(2k)$ and using the fact that for undirected graphs we can ensure in $O(m)$ time that $m=O(nk)$ \cite{NagamochiI92}.%

\paragraph{Maximal $k$-Edge Connected Subgraphs.}
For any set of vertices $C\subseteq V$, its induced subgraph $G[C]$ is a maximal $k$-edge-connected subgraph of $G$ if $G[C]$ is a $k$-edge-connected graph and no superset of $C$ has this property.\footnote{Recall that a graph is $k$-edge connected we we need to remove at least $k$ edges so that there is no path from some node $u$ to another node $v$.} Chechik~et~al. \cite{ChechikHILP17} presented deterministic algorithms that can compute {\em all} maximal $k$-edge-connected subgraphs in $k^{O(k)}m\sqrt{n}\log(n)$ time on undirected graphs and  $k^{O(k)}m\sqrt{m}\log(n)$ time on directed graphs.

Our result is mainly inspired by a part of Chechik~et~al.'s algorithms which runs some algorithm as a subroutine. Note that it is not hard to adapt their techniques to solve LocalVC in $k^{O(k)}\nu$ time. Our result is an improvement over this, and in turn implied an improved running time for computing the maximal $k$-edge-connected subgraphs.
We improve the dependency on $k$ in Chechik~et~al.'s result from $k^{O(k)}$ to $\poly(k)$.

\paragraph{Independent work by Forster and Yang \cite{ForsterY19}.} Independently from this paper (see \Cref{foot:independent_result}), Forster and Yang present
results similar to the above-mentioned results, except that 
(i) they show additional steps that lead to a better time complexity for computing the maximal $k$-edge connected subgraphs on undirected graphs, namely $O(k^4n^{3/2}\log n + km\log^2 n)$,
(ii) they show some applications in graph property testing, and
(iii) they do not consider approximation algorithms.
Inspire by their property testing results, and by using approximation LocalVC algorithms, we obtain new bounds for testing vertex- and edge-connectivity below.

\paragraph{Testing Vertex- and Edge- Connectivity.} The study of testing graph properties, initiated by Goldreich et~al. \cite{GoldreichGR98}, concerns the number of {\em queries} made to answer a question about graph properties.
In the {\em (unbounded-degree) incident-lists model} \cite{GoldreichR02,OrensteinR11}, it is assumed that there is a list $L_v$ of edges incident to each node $v$ (or lists of outgoing and incoming edges for directed graphs), and  an algorithm can make a query $q(v, i)$ for the $i^{th}$ edge in the list $L_v$ (if $i$ is bigger than the list size, the algorithm receives a special symbol in return). For any $\epsilon>0$, we say that an $m$-edge graph $G$ is {\em $\epsilon$-far} from having a property $P$ if the number of edge insertions and deletions to make $G$ satisfies $P$ is at least $\epsilon m$. {\em Testing $k$-vertex connectivity} is a problem where we want to distinguish between when $G$ is $k$-vertex connectivity and when it is $\epsilon$-far from having such property. Testing $k$-edge connectivity is defined analogously. It is assumed that the algorithm receives $n$, $\epsilon$, and $k$ in the beginning. We show the following. 
\begin{theorem}\label{thm:intro:vertex-testing}
In the unbounded-degree incident-list model, $k$-vertex (where $k < n/4$) and -edge connectivity for directed graphs can be tested in  $\ot(k/\epsilon^2)$ queries with probability at least $2/3$. %
 Further, $k$-edge connectivity for {\em simple} directed graphs can be tested in $\ot(\min\{k/\epsilon^2, 1/\epsilon^3 \}) $ queries. 
\end{theorem}
In particular, our  $\ot(k/\epsilon^2)$  bound is {\em linear} in $k$, and it can be independent of $k$ for testing $k$-edge connectivity on simple graphs. 
In the  {\em bounded-degree} incident-list model, the maximum degree $d$ is assumed to be given to the algorithm and a graph is said to be $\epsilon$-far from a property $P$ if it needs at least $\epsilon nd$ edge modifications to have such property. We show the following. 
\begin{theorem} \label{thm:intro:vertex-testing-bounded}
In the bounded-degree incident-list model, $k$-vertex (where $k < n/4$) and -edge connectivity for directed graphs can be tested in  $\ot(k/\epsilon)$ queries with probability at least $2/3$. %
 Further, $k$-edge connectivity for {\em simple} directed graphs can be tested in $\ot(\min\{k/\epsilon,  1/\epsilon^2 \}) $ queries. 
\end{theorem}

It has been open for many years whether the bounds from \cite{GoldreichR02,OrensteinR11,YoshidaI10,YoshidaI12} which are exponential in $k$ can be made polynomial (this was asked in e.g. \cite{OrensteinR11}). Forster and Yang \cite{ForsterY19} answered this using the same result as our local algorithms. The dependence on $k$ in their bounds is at least $k^3$,  
even on bounded-degree graphs. We can improve the dependence on $k$ essentially by using approximation local algorithms.

\begin{table}[]
\centering

\caption{Comparison of property testing algorithms.}
\begin{tabular}{|c|c|c|}
\hline
Unbounded-Degree                      & From              \cite{ForsterY19}                                                           & This paper                                                    \\ \hline
$k$-edge-connectivity                 & $\ot( k^4 / (\epsilon^2 \bar{d}^2) ) = \ot(k^2/\epsilon^2)$ & $\ot(k^2/(\epsilon^2 \bar{d})) = \ot(k/\epsilon^2) $    \\
$k$-edge-connectivity on simple graphs & $\ot( k^4 / (\epsilon^2 \bar{d}^2) ) = \ot(k^2/\epsilon^2)$ & $\ot(\min\{ k^2/(\bar{d}\epsilon^2), k/(\bar{d}\epsilon^3) \})$ \\
\multicolumn{1}{|l|}{}                & \multicolumn{1}{l|}{}                                       & $= \ot(\min\{k/\epsilon^2, 1/\epsilon^3 \})$                                         \\
$k$-vertex-connectivity               & $\ot( k^5 / (\epsilon^2 \bar{d}^2) ) = \ot(k^3/\epsilon^2)$ & $\ot(k^2/(\epsilon^2 \bar{d})) = \ot(k/\epsilon^2) $          \\ \hline
Bounded-Degree                        & From         \cite{ForsterY19}                                                               & This paper                                                    \\ \hline
$k$-edge-connectivity                 & $\ot(k^3/\epsilon)$                                         & $\ot(k/\epsilon)$                                             \\
$k$-edge-connectivity on simple graphs & $\ot(k^3/\epsilon)$                                         & $\ot(\min\{k/\epsilon,  1/\epsilon^2 \})$      \\
$k$-vertex-connectivity               & $\ot(k^3/\epsilon)$                                         & $\ot(k/\epsilon)$                                             \\ \hline
\end{tabular}
\label{tb:comparison}
\end{table}
 
\medskip \noindent {\em Detailed comparisons:} 
To precisely compare our bounds with the previous ones, note that there are two sub-models: (i) In the  {\em unbounded-degree} incident-list model, previous work assumes that $\bar d=m/n$ is known to the algorithm in the beginning. (ii) in the  {\em bounded-degree} incident-list model, the maximum degree $d$ is assumed to be given to the algorithm and a graph is said to be $\epsilon$-far from a property $P$ if it needs at least $\epsilon nd$ edge modifications to have such property.
Our $\ot(k/\epsilon^2)$  bound can be generalized to $\ot(k^2/(\epsilon^2 \bar d))$  bound in the unbounded-degree model. Similarly, our  $\ot(\min\{k/\epsilon^2, 1/\epsilon^3 \})$ bound can be generalized to $\ot(\min\{ k^2/(\bar{d}\epsilon^2), k/(\bar{d}\epsilon^3) \})$  bound in the unbounded-degree model.
\footnote{\Cref{thm:intro:vertex-testing} can be obtained simply from the fact that $k\leq \bar d, d \leq k/\epsilon$ can be assumed without loss of generality.}
The bounds that are exponential in $k$ by \cite{OrensteinR11,YoshidaI10,YoshidaI12} are $\tilde O((\frac{ck}{\epsilon \bar d})^{k+1})$ and $\tilde O((\frac{ck}{\epsilon d})^{k}d)$ in the unbounded- and bounded-degree models, respectively, for testing both directed $k$-vertex and -edge connectivity. The bounds that are polynomial in $k$ by Forster and Yang \cite{ForsterY19} are (i) $\tilde O(k^5/(\epsilon \bar d)^2)$ for $k$-vertex connectivity in the unbounded-degree model,  (ii) $\tilde O(k^4/(\epsilon \bar d)^2)$ for $k$-edge connectivity in the unbounded-degree model, and (iii) $\tilde O(k^3/\epsilon)$ for both $k$-vertex and -edge connectivity in the bounded-degree model. Table \ref{tb:comparison} details comparisons between our results, and those from \cite{ForsterY19}.

%% file: prelim.tex
\section{Preliminaries}\label{sec:prelim}

Let $G=(V,E)$ be a \emph{directed} graph. For any $S,T\subseteq V$,
let $E(S,T)=\{(u,v)\mid u\in S,v\in T\}$. For each vertex $u$, we
let $\deg^{\out}(u)$ denote the out-degree of $u$ respectively.
For a set $S\subseteq V$, the \emph{out-volume} of $S$ is $\vol^{\out}(S)=\sum_{u\in S}\deg^{\out}u$.
The a set of \emph{out-neighbors} of $S$ is $N^{\out}(S)=\{v\mid(u,v)\in E(S,V-S)\}$.
We can define in-degree $\deg^{\inn}(u)$, in-volume $\vol^{\inn}(S)$,
and a set of in-neighbors $N^{\inn}(S)$ analogously. We add a subscript
$G$ to the notation when it is not clear which graph we are referring
to. 

We say that $(L,S,R)$ is a \emph{separation triple} of $G$ if $L,S,R$
partition $V$ where $L,R\neq\emptyset$, and $E(L,R)=\emptyset$.
We also say that $S$ is a \emph{vertex cut} of $G$ of size $|S|$.
$S$ is an \emph{$st$-vertex cut} if $s\in L$ and $t\in R$. We
say that $s$ and $t$ is $k$-connected (or $k$-vertex-connected)
if there is no $st$-vertex cut of size less than $k$. $G$ is $k$-connected
if $s$ and $t$ is $k$-connected for every pair $s,t\in V$.

%% file: ec_new.tex
\section{Local edge connectivity}
\label{sec:localEC} 

In this section, we show a local algorithm for detecting an edge cut
of size $k$ and volume $\nu$ containing some seed node in $O(\nu
k^{2})$ time. The algorithm accesses (i.e. make queries
for) $O(\nu k)$ edges (this is needed for our property
testing results). Roughly, it outputs either a ``small'' cut or a
symbol ``$\bot$'' indicating that there is no small cut containing
$x$. The algorithm makes errors (i.e. it might be wrong when it
outputs $\bot$) %
with probability at most $1/4$. By standard arguments, we can make
the error probability arbitrarily small. 
Both the algorithm description and the analysis are very simple.

\begin{theorem}
\label{thm:localEC_approx_new}
There exists the following randomized algorithm. It takes as inputs, 
\begin{itemize}[noitemsep,nolistsep]
\item a pointer to  an adjacency list representing  an $n$-vertex $m$-edge graph
$G=(V,E)$, 
\item a seed vertex $x \in V$, 
\item a volume parameter (positive integer) $\nu$, 
\item a cut-size parameter (positive integer) $k$, and 
\item a slack parameter (non-negative integer) $\gap$, where
\end{itemize}
\begin{align} \label{eq:precon_ec}
k\geq 1, \quad \nu > k, \quad  \gap \leq k \quad \mbox{ and } \quad \nu < m(\gap+1)/(130k).
\end{align}
It accesses (i.e. makes queries for) $O(\nu k /(\gap+1))$ edges and runs in $O(\nu k^2/ (\gap+1))$ time. It then outputs 
in the following manner. 
\begin{itemize}[noitemsep,nolistsep]
\item If there exists a vertex-set $S'$ such that $S' \ni x, \vol^{\out}(S')
  \leq \nu, $ and $ |E(S',V-S')| < k$, then with probability at least $3/4$, the algorithm
  outputs $S$ a non-empty vertex-set $S\subsetneq V$
  such that $|E(S,V-S)| < k+\gap$ and $\vol^{\out}(S) \leq 130\nu k/(\gap+1)$  (otherwise it outputs $\bot$).  
\item Otherwise (i.e., no such $S'$ exists), the algorithm outputs either
  the set $S$ as above or $\bot$. 
\end{itemize}
\end{theorem}

In particular, we obtain exact and $(1+\epsilon)$-approximate local
algorithms for \Cref{thm:localEC_approx_new} as follows. We set $\gap
= 0$ in \Cref{thm:localEC_approx_new} for the exact local algorithm.
For $(1+\epsilon)$-approximation, we set $\gap = \lfloor \epsilon k
\rfloor.$

\begin{corollary}
\label{cor:localEC_exact_new}
There exists the following randomized algorithm. It takes as the same inputs as
in \Cref{thm:localEC_approx_new} where $\gap = 0$.  
It accesses (i.e. makes queries for) $O(\nu k)$ edges and
runs in $O(\nu k^2)$ time.  It then outputs 
in the following manner.  
\begin{itemize}[noitemsep,nolistsep]
\item If there exists a vertex-set $S'$ such that $S' \ni x, \vol^{\out}(S')
  \leq \nu, $ and $ |E(S',V-S')| < k$, then with probability at least $3/4$, the algorithm
  outputs $S$ a non-empty vertex-set $S\subsetneq V$
  such that $|E(S,V-S)| < k $  (otherwise it outputs $\bot$).  
\item Otherwise (i.e., no such $S'$ exists), the algorithm outputs either
  the set $S$ as above or $\bot$. 
\end{itemize}
\end{corollary}
 
\begin{corollary}
\label{cor:localEC_approx_new}
There exists the following randomized algorithm. It takes as the same inputs as
in \Cref{thm:localEC_approx_new} with additional parameter $\epsilon
\in (0,1]$ where $\gap$ is set to be equal to $ \lfloor \epsilon k \rfloor.$  
It accesses (i.e. makes queries for) $O(\nu /\epsilon)$ edges and
runs in $O(\nu k/ \epsilon)$ time.  It then outputs 
in the following manner.  
\begin{itemize}[noitemsep,nolistsep]
\item If there exists a vertex-set $S'$ such that $S' \ni x, \vol^{\out}(S')
  \leq \nu, $ and $ |E(S',V-S')| < k$, then with probability at least $3/4$, the algorithm
  outputs $S$ a non-empty vertex-set $S\subsetneq V$
  such that $|E(S,V-S)| < \lfloor (1+\epsilon) k \rfloor $  (otherwise it outputs $\bot$).  
\item Otherwise (i.e., no such $S'$ exists), the algorithm outputs either
  the set $S$ as above or $\bot$. 
\end{itemize}
\end{corollary}
\begin{proof}
The results follow from \Cref{thm:localEC_approx_new} where we set
$\gap = \lfloor \epsilon k \rfloor$, and the following fact: 
\begin{align} \label{eq:gap-to-approx}
k+\gap = k + \lfloor \epsilon k \rfloor  =\lfloor k \rfloor + \lfloor
  \epsilon k \rfloor  \leq \lfloor k+\epsilon k\rfloor = \lfloor (1+\epsilon) k\rfloor.
\end{align}\qedhere
\end{proof}
 
\begin{algorithm}
\SetKwFor{RepTimes}{repeat}{times}{end}
\KwIn{$\text{Seed node } x \in V, \text{target volume } \nu \geq k,
  \text{target cut-size } k \geq 1,  \text{slack }\gap \leq k.$} 
\KwOut{a vertex-set $S$ or the symbol $\perp$ with as in \Cref{thm:localEC_approx_new}.}     
\BlankLine

\RepTimes{$k + \gap$}
{
$y \gets \mathtt{NIL}.$ \;
Grow a BFS tree $T$ starting from $x$ (where every edges point towards leaves) as follows: \;
\While{\normalfont{the BFS algorithm still has an edge $e = (a,b)$ to explore}} 
{
  \If{$e$ \normalfont{is} not marked} {
       Mark $e$.\;  
       \lIf{\normalfont{the} algorithm marks $\geq 128\nu k/(\gap+1)$
         edges  \label{line:earlybot}}{
       	\Return{$\bot$. } }  
       With probability $(\gap+1)/(8\nu)$, set $y \gets a$ and \textbf{break the while-loop}.\;
    }
}
\lIf{$y = \mathtt{NIL}$} { 
  \Return $V(T)$. 
} \lElse  {Reverse the direction of edges in the path $P_{xy}$ in $T$ from $x$ to $y$.}
}
\Return{$\bot$.}

\caption{$\localEC(x,\nu,k,\gap)$\label{alg:local_ec_new}}
\end{algorithm}

The algorithm for \Cref{thm:localEC_approx_new} is described in
\Cref{alg:local_ec_new}.  
Roughly, in each iteration of the
repeat-loop the algorithm runs a standard breadth-first search (BFS)
algorithm, but randomly stop before finishing the BFS algorithm; this
means we break the while-loop in \Cref{alg:local_ec_new} . (If the BFS algorithm finishes first, it outputs all nodes found in such iteration.) If an edge $(a,b)$ is the last edge explored by the algorithm before the random stop happens, we flip the direction of all edges on the unique path from $x$ to $b$ in the BFS tree. We repeat this for $k+\gap$ iterations.

In the analysis below, we assume that every node accessed by
\Cref{alg:local_ec_new} has degree at least $k$. Otherwise, an
edge-cut is found and we are done.  We start with the following important observation.

\begin{lemma} \label{lem:reverse_cutsize_new} Let $S\subset V$ be any set where $x\in S$.
Let $P_{xy}$ be a path from $x$ to $y$. Suppose we reverse the
direction of edges in $P_{xy}$. Then, we have $|E(S,V-S)|$ and $\vol^{\out}(S)$
are both decreased exactly by one if $y\notin S$. Otherwise, $|E(S,V-S)|$
and $\vol^{\out}(S)$ stay the same.
\end{lemma}
\begin{proof}
We fix the set $S$ and the path $P_{xy}$ where $x \in S$. If $y \notin
S$, then $P_{xy}$ crosses the edges between $S$ and $V-S$ back and forth
so that the number of out-crossing
edges (i.e., from $S$ to $V-S$) is one plus the number of in-crossing
edges (i.e., from $V-S$ to $S$). Therefore, reversing the direction of
edges in $P_{xy}$ decreases both $|E(S,V-S)|$ and $\vol^{\out}(S)$
exactly by one.  If $y \in S$, then the number out-crossing edges is
equal to the number of in-crossing edges. Thus, reversing the
direction of edges in $P_{xy}$ does not change $|E(S,V-S)|$ and
$\vol^{\out}(S)$. 
\end{proof}

It is easy to see \Cref{alg:local_ec_new} accesses  $O(\nu
k/(\gap+1))$ edges, and runs in time $O(\nu k^2/(\gap+1))$. 

\begin{lemma} \label{lem:running-time-main}
\Cref{alg:local_ec_new} runs in time $O(\nu k^2/(\gap+1))$,
and  accesses $O(\nu k/(\gap+1))$ edges. 
\end{lemma}
\begin{proof}
By \label{line:earlybot}, the algorithm accesses at most $\lceil 128
\nu k /(\gap +1) \rceil = O(\nu k/(\gap+1)) $ edges.   The running
time follows since BFS runs linear in number of edge accesses, and we repeat at most $k+\gap \leq 2k$ iterations. 
\end{proof}

The following two lemmas prove the correctness of \Cref{alg:local_ec_new}.
\begin{lemma}
If a vertex set $S$ is returned, then  $\emptyset \neq S \subset V,
|E(S,V-S)|< k + \gap$ and $\vol^{\out}(S) \leq 130\nu k/(\gap+1)$. 
\end{lemma} 
\begin{proof}
If $S$ is returned, then the BFS tree $T$ gets stuck at $S=V(T)$.
That is, at the end of the algorithm, we have  $|E(S,V-S)|=0$ and $
\vol^{\out}(S) \leq 128\nu k/(\gap+1) \stackrel{(\ref{eq:precon_ec})}  < m .$  We first show that $S \neq
\emptyset$ and $S \neq V$. Note that by design either $S$ contains a
single node or $x \in S$. The fact that $S \neq V$ follows since
$\vol^{\out}(S)   < m$ and  BFS tree $T$ gets stuck with a spanning tree, $V(T) = V$, only if all edges are
marked.   The algorithm has reversed strictly less than $k+\gap $ many paths $P_{xy}$ because
the algorithm did not reverse a path in the iteration that $S$ is returned. Therefore,
\Cref{lem:reverse_cutsize_new} implies that, initially, $|E(S,V-S)|< k
+ \gap,$ and $\vol^{\out}(S)  < 128\nu k/(\gap+1)  + k+\gap \leq 130\nu k/(\gap+1) .$

\end{proof} %

\begin{lemma}
If there is $S$  where $x \in S, |E(S,V-S)|<k$ and
$\vol^{\out}(S)\le\nu$, then $\bot$ is returned with probability at
most $1/4$. 
\end{lemma}
\begin{proof}

Let $\tau = 128\nu k/(\gap+1)$, and $R_{\bot}$ be an event that $\bot$ is
returned.  For the purpose of analysis, let $X$ be a random
variable denoting the total number of edges that the algorithm marked
assuming line~\earlybotline~ is ignored,  which is the same as
the total number of edges accessed by the algorithm.   Our goal is to show that $\Pr[R_{\bot}] \leq 1/4$. It is
enough to show that $\Pr[X \geq \tau] \leq 1/8$ and $\Pr[R_{\bot} \mid
X < \tau] \leq 1/8$. If this is true, then we have the following. 
\begin{align*}
\Pr[R_{\bot}] &= \Pr[R_{\bot} \cap (X < \tau) ] + \Pr[R_{\bot} \cap (X \geq
  \tau)]\\ 
&= \Pr[R_{\bot} \mid  X < \tau]\Pr[ X < \tau]  + \Pr[R_{\bot} \mid X
  \geq \tau] \Pr[X  \geq \tau]\\
&\leq  \Pr[R_{\bot} \mid  X < \tau] + \Pr[X  \geq \tau] \\
& \leq 1/8+1/8 = 1/4.
\end{align*}
The first two equalities follow from conditional probability. The
third inequality follows since $\Pr[X < \tau] \leq 1, $ and if $X
\geq \tau$, the the algorithm always outputs $\bot$. 

We show that $\Pr[X \geq \tau] \leq 1/8$. In other words, we show  the
probability that the algorithm returns $\bot$ at
line~\earlybotline~ is at most $1/8$. 
We first show $E[X] \leq 16 \nu k/(\gap+1)$.   For $i \in \{1,\ldots,
k+\gap\}$, let $X_i$ be a random variable denoting  the number of edges that are marked at the
$i$th iteration of the repeat-loop. Then $E[X_i] = 8\nu/(\gap+1)$. Let
$X = \sum_{i =1}^{k+\gap}X_i$. By linearity of expectation,   we have 
\begin{align} \label{eq:expected-read-edges}
E[X] = \sum_{i=1}^{k+\gap}E[X_i] = (k+\gap)(8\nu)/(\gap+1)  \leq 16\nu k/(\gap+1).
\end{align}

It remains to show that $\Pr[X <  \tau
] \geq 1-1/8$. By \Cref{eq:expected-read-edges}, $E[X] \leq 16\nu
k/(\gap+1)$, so $8\cdot E[X] \leq 128\nu k/(\gap+1) = \tau.$
By Markov's inequality,  $\Pr[X < 8\cdot E[X]] \geq 1-1/8$. Therefore,
$\Pr[X < \tau] \geq 1-1/8$, and  we have $\Pr[X \geq \tau] \leq 1/8$.

Next, we show that $\Pr[R_{\bot} \mid X < \tau] \leq 1/8$. In other
words, given that $\bot$ is not returned at line~\earlybotline,
we compute the probability that $\bot$ is returned at the last
line. Suppose that no set $S'$ is returned before the last
iteration. We will show that $\Pr[\bar{R_{\bot}} \mid X < \tau] \geq 1-1/8$. It suffices to show that the number of stops at edge
$(a,b)$ where $a \in S$ (i.e., an edge $(a,b) \in E(S,V)$) is at most $\gap$ after $k+\gap$ iterations with probability at least $1-1/8$. If this is
true, and such an event happens,
then there are at least $k-1$ iterations where we stop at edges
 in the set $E(V-S,V)$, and we do not stop at any edge in the set
 $E(S,V)$ %
 in the final iteration. \Cref{lem:reverse_cutsize_new}
implies that $|E(S,V-S)| = 0$ at the beginning of the last iteration.  Therefore, the the BFS tree $T$ will
return the set $V(T)$ since we do not stop at any edge in $E(S,V)$ in
the final iteration.  

Let $Z$ be a random variable denoting the number of stops at
an edge in $E(S,V)$ after $k+\gap$ iterations.  We show
that $\Pr[Z \leq \gap] \geq 1-1/8$. Since we never unmark edges, and
we stop with probability $(\gap+1)/(8\nu)$ for each first visited edge
by  linearity of expectation, we have $E[Z] \leq \sum_{e \in E(S,V)}
(\gap+1)/(8\nu) \leq  \nu (\gap+1)/(8\nu)  = (\gap+1)/8$.   Therefore,
by Markov's inequality, we have $P(Z < 8E[Z]) \geq 1-1/8$. Note that
$Z < 8E[Z] \leq \gap$. Hence, $P(Z \leq \gap) \geq 1-1/8$ as $\gap$ is
an integer.  
\end{proof}

\section{Local vertex connectivity}

\label{sec:localVC}

In this section, we show the vertex cut variant of the local algorithms
from \Cref{sec:localEC}.

\begin{theorem}
\label{thm:localVC_gap}
There exists the following randomized algorithm. It takes as inputs, 
\begin{itemize}[noitemsep,nolistsep]
\item a pointer to  an adjacency list representing  an $n$-vertex $m$-edge graph
$G=(V,E)$, 
\item a seed vertex $x \in V$, 
\item a volume parameter (positive integer) $\nu$, 
\item a cut-size parameter (positive integer) $k$, and 
\item a slack parameter (non-negative integer) $\gap$, where
\end{itemize}
\begin{align} \label{eq:precon_vc}
k\geq 1, \quad \nu > k, \quad  \gap \leq k \quad \mbox{ and } \quad \nu < m(\gap+1)/(8320k).
\end{align}
It accesses (i.e. makes queries for) $O(\nu k /(\gap+1))$ edges and runs in $O(\nu k^2/ (\gap+1))$ time. It then outputs 
in the following manner. 
\begin{itemize}[noitemsep,nolistsep]
\item If there exists a separation triple $(L',S',R')$ such that $L' \ni x, \vol^{\out}(L')
  \leq \nu, $ and $|S'| < k$, then with probability at least $3/4$, the algorithm
  outputs a vertex-cut of size at most $k+\gap$ (otherwise it outputs $\bot$).  
\item Otherwise (i.e., no such separation triple $(L',S',R')$ exists), the algorithm outputs either
  a vertex-cut of size at most $k+\gap$ or $\bot$. 
\end{itemize}
\end{theorem}

We obtain exact and $(1+\epsilon)$-approximate local algorithms for
\Cref{thm:localVC_gap} by setting $\gap = 0$ and $\gap =  \lfloor
\epsilon k \rfloor,$ respectively. 

\begin{corollary}
\label{cor:localVC_exact}
There exists the following randomized algorithm. It takes as the same inputs as
in \Cref{thm:localVC_gap} where $\gap$ is set to be equal to $0$.   
It accesses (i.e. makes queries for) $O(\nu k)$ edges and
runs in $O(\nu k^2 )$ time.  It then outputs 
in the following manner.  
\begin{itemize}[noitemsep,nolistsep]
\item  If there exists a separation triple $(L',S',R')$ such that $L'
  \ni x, \vol^{\out}(L') \leq \nu, $ and $|S'| < k$, then with probability at least $3/4$, the algorithm
  outputs a vertex-cut of size at most $k$ (otherwise it outputs $\bot$).  
\item Otherwise (i.e., no such $(L',S',R')$ exists), the algorithm
  outputs either  a vertex-cut of size at most $ k $  or $\bot$. 
\end{itemize}
\end{corollary}

\begin{corollary}
\label{cor:localVC_approx}
There exists the following randomized algorithm. It takes as the same inputs as
in \Cref{thm:localVC_gap} with additional parameter $\epsilon
\in (0,1]$ where $\gap$ is set to be equal to $ \lfloor \epsilon k \rfloor.$  
It accesses (i.e. makes queries for) $O(\nu /\epsilon)$ edges and
runs in $O(\nu k/ \epsilon)$ time.  It then outputs 
in the following manner.  
\begin{itemize}[noitemsep,nolistsep]
\item  If there exists a separation triple $(L',S',R')$ such that $L' \ni x, \vol^{\out}(L')
  \leq \nu, $ and $|S'| <  k$, then with probability at least $3/4$, the algorithm
  outputs a vertex-cut of size at most $\lfloor (1+\epsilon) k \rfloor$ (otherwise it outputs $\bot$).  
\item Otherwise (i.e., no such $(L',S',R')$ exists), the algorithm
  outputs either  a vertex-cut of size at most $\lfloor (1+\epsilon) k
  \rfloor$  or $\bot$. 
\end{itemize}
\end{corollary}
\begin{proof}
The results follow from \Cref{thm:localVC_gap} where we set
$\gap = \lfloor \epsilon k \rfloor$, and \Cref{eq:gap-to-approx}.
\end{proof}

To prove \Cref{thm:localVC_gap}, in \Cref{sec:reduc} we first
reduce the problem to the edge version of the problem using the well-known
reduction (e.g. \cite{Even75,HenzingerRG00}) and then in \Cref{sec:reduc_apply}
we plug the algorithm from \Cref{thm:localEC_approx_new} into the reduction.

\subsection{Reducing from vertex to edge connectivity}

\label{sec:reduc}

Given a directed $n$-vertex $m$-edge graph $G=(V,E)$ and a vertex
$x\in V$, let $G'=(V',E')$ be an $n'$-vertex $m'$-edge graph defined
as follows. We call $G'$ the \emph{split graph} w.r.t. $x$. For each $v\in V-\{x\}$,
we add vertices $v_{\inn}$ and $v_{\out}$ to $V'$ and a directed
edge $(v_{\inn},v_{\out})$ to $E'$. We also add $x$
to $V'$ and we denote $x_{\inn}=x_{\out}=x$ for convenience. For
each edge $(u,v)\in E$, we add $(u_{\out},v_{\inn})$ to $E'$. Suppose,
for convenience, that the minimum out-degree of vertices in $G$ is
$1$. The following two lemmas draw connections between $G$ and $G'$.

\begin{lemma}
\label{lem:completeness}Let $(L,S,R)$ be a separation triple in
$G$ where $S=N_{G}^{\out}(L)$. Let $L'=\{v_{\inn},v_{\out}\mid v\in L\}\cup\{v_{\inn}\mid v\in S\}$
be a set of vertices in $G'$. Then $|E_{G'}(L',V'-L')|=|S|$ and
$\vol_{G}^{\out}(L)\le\vol_{G'}^{\out}(L')\le2\vol_{G}^{\out}(L)$.
\end{lemma}

\begin{proof}
As $S=N_{G}^{\out}(L)$, we have $|E_{G'}(L',V'-L')|=|\{(v_{\inn},v_{\out})\mid v\in S\}|=|S|$.
Also, $\vol_{G'}^{\out}(L')=\vol_{G}^{\out}(L)+|L|+|S|\le2\vol_{G}^{\out}(L)$
because every vertex in $G$ has out-degree at least $1$ and $S=N_{G}^{\out}(L)$.
\end{proof}
\begin{figure}
	\begin{centering}
		\includegraphics[viewport=180bp 200bp 550bp 470bp,clip,width=0.35\paperwidth]{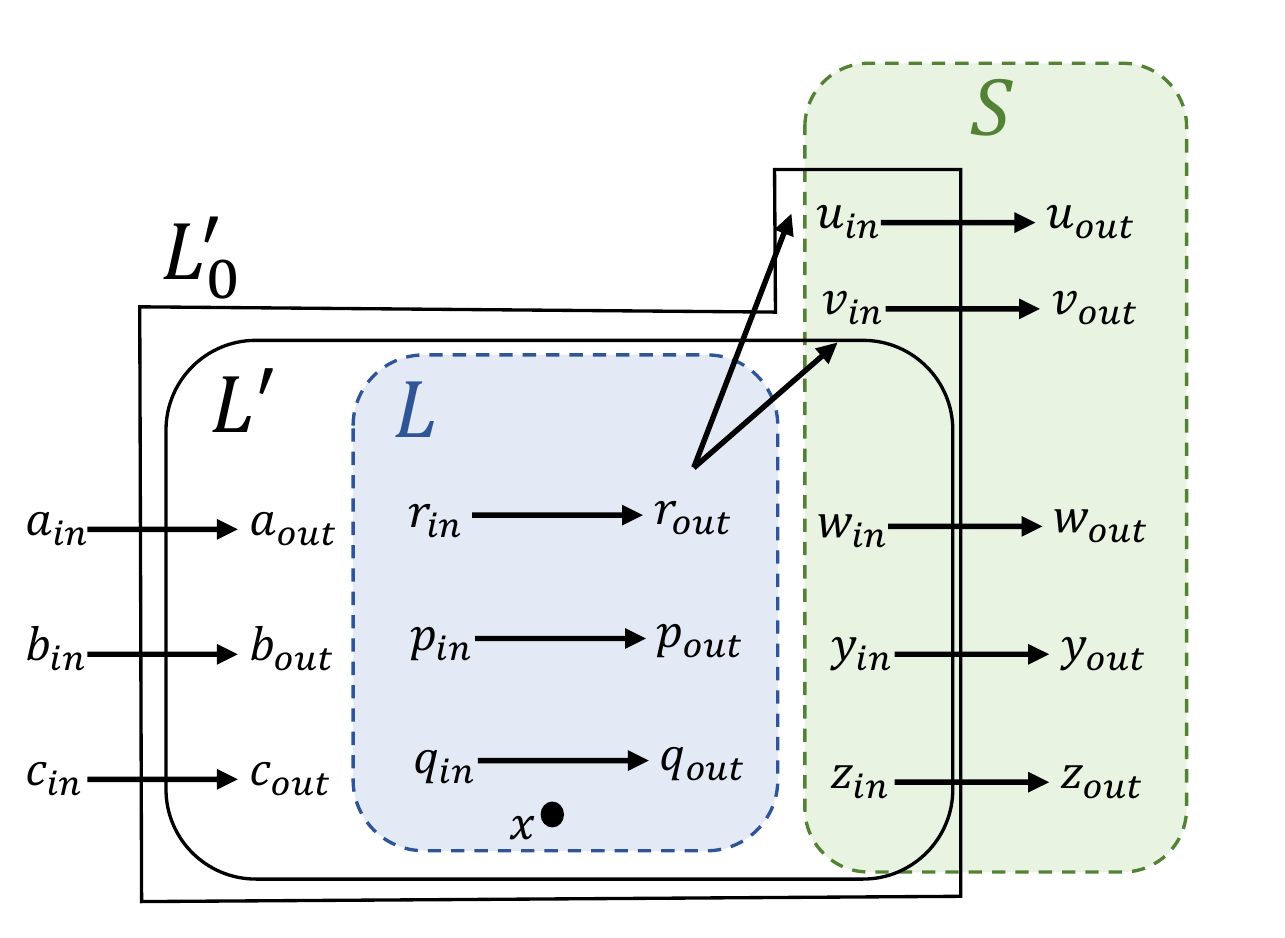}
		\par\end{centering}
	\caption{Construction of a separation triple $(L,S,R)$ in $G$ from an edge cut $L'$
		in $G'$. Most edges are omitted. From this example, $L'_{0}=L'\cup\{u_{\protect\inn},v_{\protect\inn}$\},
		$L=\{r,p,q,x\}$, $S=\{u,v,w,y,z\}$.\label{fig:reduc convert back}}
\end{figure}

\begin{lemma}
\label{lem:soundness}Let $L'\ni x$ be a set of vertices of $G'$.
Then, there is a set of vertices $L$ in $G$ such that $\vol_{G}^{\out}(L)\le2\vol_{G'}^{\out}(L')$
and $|S|\le|E_{G'}(L',V'-L')|$ where $S=N_{G}(L)$. Given $L'$,
$L$ can be constructed in $O(\vol_{G'}^{\out}(L'))$ time.

Let $R=V-(L\cup S)$. We have $R\neq\emptyset$, i.e. $(L,S,R)$ is
a separation triple if $\vol_{G'}^{\out}(L')\le m'/32$ and $|E_{G'}(L',V'-L')|\le n/2$.

\end{lemma}

\begin{proof}
First of all, note that if there is $v$ where $v_{\out}\in L'$ and
$\deg_{G'}^{\out}(v_{\out})\le|E_{G'}(L',V'-L')|$, then we can return
$L=\{v\}$ and $S=N_{G}^{\out}(\{v\})$ and we are done. So from now,
we assume that $\deg_{G'}^{\out}(v_{\out})>|E_{G'}(L',V'-L')|$.

By the structure of $G'$, observe that there are sets $S_{1},S_{2}\subseteq V$
be such that 
\[
E_{G'}(L',V'-L')=\{(v_{\inn},v_{\out})\mid v\in S_{1}\}\cup\{(u_{\out},v_{\inn})\mid v\in S_{2}\}
\]
Let $L'_{0}=L'\cup\{v_{\inn}\mid v\in S_{2}\}$. See \Cref{fig:reduc
	convert back} for illustration. So there is a set $S\subset V$ where
\[
E_{G'}(L'_{0},V'-L'_{0})=\{(v_{\inn},v_{\out})\mid v\in S\}.
\]
We have $|S|=|E_{G'}(L'_{0},V'-L'_{0})|\le|E_{G'}(L',V'-L')|$ because
for each $v\in S_{2}$, $\deg^{\out}(v_{\inn})\le1\le\deg^{\inn}(v_{\inn})$.
Also, $\vol_{G'}^{\out}(L'_{0})\le\vol_{G'}^{\out}(L')+|E_{G'}(L',V'-L')|\le2\vol_{G'}^{\out}(L')$.

Let $L=\{v\mid v_{\inn},v_{\out}\in L'_{0}\}$. Note that $L\cap S=\emptyset$.
See \Cref{fig:reduc convert back} for illustration. Observe that
$x\in L$ because $x_{\out}=x_{\inn}$. Moreover, $N_{G}(L)=S$ and
$\vol_{G}^{\out}(L)\le\vol_{G'}^{\out}(L'_{0})\le2\vol_{G'}^{\out}(L')$.
$L$ be can constructed in time $O(|\{v_{\out}\in L'\}|)=O(\vol_{G'}^{\out}(L'))$
because the minimum out-degree of vertices in $G$ is $1$.

For the second statement, observe that $R=V-(L\cup S)=\{v\mid v_{\inn}\notin L'_{0}\}$.
Let $V'_{\inn}=\{v_{\inn}\in V'\}\cup\{x\}$ and $V'_{\out}=\{v_{\out}\in V'\}-\{x\}$.
Let $k'=|E_{G'}(L'_{0},V'-L'_{0})|$. Suppose for contradiction that
$R=\emptyset$. We claim that 
\[
|V'-L'_{0}|=|V'_{\out}-L'_{0}|=k'.
\]
This is because $L'_{0}\supseteq V'_{\inn}$, $V'-L'_{0}\subseteq V'_{\out}$,
and $E_{G'}(L'_{0},V'-L'_{0})$ only contains edges of the form $(v_{\inn},v_{\out})$.
Now, there are two cases. If $m'\ge4n'k'$, then we have 
\begin{align*}
m' & =\vol_{G'}^{\out}(L'_{0})+\vol_{G'}^{\out}(V'-L'_{0})\\
& \le2\vol_{G'}^{\out}(L')+n'|V'_{\out}-L'_{0}|\\
& \le2\cdot m'/32+n'k'\\
& \le m'/16+m'/4<m'
\end{align*}
which is a contradiction. Otherwise, we have $m'<4n'k'$. Note that
$n'<2n$ by the construction of $G'$ and so $m'<8nk'$. Hence, we
have 
\[
\vol_{G'}^{\out}(L'_{0})\ge|L'_{0}\cap V'_{\out}|k'\ge(n-k')k'\ge nk'/2>m'/16
\]
which contradicts $\vol_{G'}^{\out}(L'_{0})\le2\vol_{G'}^{\out}(L')\le2\cdot m'/32$. 
\end{proof}

\subsection{Proof of \Cref{thm:localVC_gap}}

\label{sec:reduc_apply}
Given an $n$-vertex $m$-edge $G=(V,E)$
represented as adjacency lists, a vertex $x\in V$ and parameters
$\nu,k,\gap$ from \Cref{thm:localVC_gap} where $\nu\le (\gap+1) m/(8320k)$
and $k\le n/4$, we will work on the split graph $G'$ with $n'$-vertex
$m'$-edge as described in \Cref{sec:reduc}. The adjacency list
of $G'$ can be created ``on the fly''. Let $\localEC(x',\nu',k',\gap')$
denote the algorithm from \Cref{thm:localEC_approx_new} with parameters
$x',\nu',k',\gap'$. We run $\localEC(x,2\nu,k,\gap)$ on
$G'=(V',E')$ in time $O(\nu k^2/(\gap+1))$. Note that $2\nu\le
(\gap+1) m/(8k)\le (\gap+1) m'/(8k)$
as required by \Cref{thm:localEC_approx_new}.

We show that if there exists a separation triple $(L,S,R)$ in $G$
where $L \ni x, |S| < k, $ and $ \vol^{\out}_G(L) \leq \nu$, then
$\localEC(x,2\nu,k,\gap)$ outputs $\bot$ with probability at most
$1/4$.  By \Cref{lem:completeness}, there exists $L'$ in $G'$ such
that $L' \ni x, |E_{G'}(L',V-L')| < k, $ and $ \vol_{G'}(L') \leq 2
\vol_G(L) \leq 2\nu$. Therefore, by \Cref{thm:localEC_approx_new},
$\localEC(x,2\nu,k,\gap)$ returns $\bot$ with probability at most
$1/4$.

If, in $G'$, $\localEC(x,2\nu,k,\gap)$ returns $L'$, then by
\Cref{thm:localEC_approx_new} we have  $L'\ni x,
|E_{G'}(L',V'-L')|< k + \gap$ and
$\vol_{G'}^{\out}(L')\le 260\nu k /\gap$.  It remains to show that  we can construct $L
\subset V$ in $G$ such that $N^{\out}_G(L)$ is a vertex-cut, and
$|N^{\out}_G(L)| < k+\gap$. By \Cref{lem:soundness},
we can obtain in $O(\nu k/(\gap+1))$ time and two sets $L$ and $S=N^{\out}(L)$
where $|S|< k+\gap$. Let $R=V-L\cup S$. As $\vol_{G'}^{\out}(L')\le
260\nu k/(\gap + 1) \stackrel{(\ref{eq:precon_vc})} \le m'/32$
and $k+\gap \le 2k \leq n/2$, we have that
$(L,S,R)$ is a separation triple by \Cref{lem:soundness}. That is, $S=
N^{\out}_G(L)$ is a vertex cut. %

%% file: vc_new.tex
\section{Vertex connectivity}\label{sec:vertex_con_global}

In this section, we show the first near-linear time algorithm for
checking of $k$-connectivity for any $k=\tilde{O}(1)$ in both undirected
and directed graphs. 
\begin{theorem}
\label{thm:VC_undir}There is a randomized (Monte Carlo) algorithm
that takes as input an undirected graph $G$ and a cut-size parameter
$k$ and an accuracy parameter $\epsilon\in(0,1]$, and in time $\tilde{O}(m+nk^{2}/\epsilon)$
either outputs a vertex cut of size less than $\left\lfloor (1+\epsilon)k\right\rfloor $
or declares that $G$ is $k$-connected w.h.p. By setting $\epsilon<1/k$,
the same algorithm decides (exact) $k$-vertex-connectivity of $G$ in $\tilde{O}(m+nk^{3})$
time.
\end{theorem}
By combining with  the state-of-the-art  algorithms for undirected graph, we obtain the the following.
\begin{corollary}
\label{cor:VC_undir}There is a randomized (Monte Carlo) algorithm
that takes as input an undirected graph $G$ and a cut-size parameter
$k$ and an accuracy parameter $\epsilon\in(0,1]$, and in time $\tilde{O}( m+ \poly(1/\epsilon) \min\{nk^{2}, n^{5/3+o(1)} k^{2/3}, n^{3+o(1)}/k, n^{\omega} \})$
either outputs a vertex cut of size less than $\left\lfloor (1+\epsilon)k\right\rfloor $
or declares that $G$ is $k$-connected w.h.p. For exact vertex connectivity,  there is a randomized  (Monte Carlo) algorithm for exact $k$-vertex-connectivity of $G$ in $\tilde{O}(m+  \min\{nk^{3}, n^2k, n^{\omega} + nk^{\omega} \} )$ time.
\end{corollary}
\begin{proof}
For approximate vertex connectivity, Nanongkai et al. \cite{NanongkaiSY19} (Theorem 1.2) present  $\tilde{O}( m+ \poly(1/\epsilon) \min\{k^{4/3}n^{4/3}, n^{5/3+o(1)} k^{2/3}, n^{3+o(1)}/k, n^{\omega} \})$-time algorithm.  By  \Cref{thm:VC_undir}, we have the $\tilde{O}(m+nk^{2}/\epsilon)$-time algorithm.  Combining  both algorithms, the term $k^{4/3}n^{4/3}$ is replaced by $nk^2$.  For exact vertex connectivity,  we combine the running time in \Cref{thm:VC_undir} with the $\ot(\min\{m+n^2k,  n^{\omega} + nk^{\omega} \})$-time algorithm, which is given  by Henzinger, Rao and Gabow \cite{HenzingerRG00}, and  Linial, Lov{\'a}sz and Wigderson \cite{LinialLW88}.
\end{proof}
Now we present the new results for directed graph. 
\begin{theorem}
\label{thm:VC_dir}There is a randomized (Monte Carlo) algorithm that
takes as input a directed graph $G$ and a cut-size parameter $k$
and an accuracy parameter $\epsilon\in(0,1]$, and in time \\
$\tilde{O}(\min\{mk/\epsilon ,  \poly(1/\epsilon) n^{2+o(1)}\sqrt{k}\})$
either outputs a vertex cut of size less than $\left\lfloor (1+\epsilon)k\right\rfloor $
or declares that $G$ is $k$-connected w.h.p.  For exact vertex connectivity,  there is a randomized  (Monte Carlo) algorithm for exact $k$-vertex-connectivity of $G$ in $\tilde{O}(  \min\{ mk^2, k^3n + k^{3/2}m^{1/2}n \} )$ time.
\end{theorem}
Similarly, by combining with  the state-of-the-art algorithms for directed graph, we obtain the the following.
\begin{corollary}
\label{cor:VC_dir}There is a randomized (Monte Carlo) algorithm
that takes as input an undirected graph $G$ and a cut-size parameter
$k$ and an accuracy parameter $\epsilon\in(0,1]$, and in time $\tilde{O}( \poly(1/\epsilon) \min\{mk, m n^{2/3+o(1)}/k^{1/3}, k^{1/2} n^{2+o(1)}, n^{7/3+o(1)}/k^{1/6}, n^{3+o(1)}/k, n^{\omega} \})$
either outputs a vertex cut of size less than $\left\lfloor (1+\epsilon)k\right\rfloor $
or declares that $G$ is $k$-connected w.h.p. For exact vertex connectivity,  there is a randomized  (Monte Carlo) algorithm for exact $k$-vertex-connectivity of $G$ in $\ot(\min\{mk^2, k^3n + k^{3/2}m^{1/2}n, mn, n^{\omega} + nk^{\omega}  \})$ time.
\end{corollary}
\begin{proof}
For approximate vertex connectivity, Nanongkai et al. \cite{NanongkaiSY19} (Theorem 1.2) present  $\tilde{O}(\poly(1/\epsilon) \min\{m^{4/3}, nm^{2/3}k^{1/2}, m n^{2/3+o(1)}/k^{1/3}, n^{7/3+o(1)}/k^{1/6}, n^{3+o(1)}/k, n^{\omega} \})$-time algorithm.  By  \Cref{thm:VC_dir}, we have the $\tilde{O}( \min\{ mk^2, k^3n + k^{3/2}m^{1/2}n \} )$-time algorithm. Combining both algorithms, the terms  $m^{4/3}$ and $nm^{2/3}k^{1/2}$ are subsumed.  For exact vertex connectivity,  we combine the running time in \Cref{thm:VC_dir} with the $\ot(\min\{ mn,  n^{\omega} + nk^{\omega} \})$-time algorithm, which is given  by Henzinger, Rao and Gabow \cite{HenzingerRG00}, and by Cheriyan and Reif \cite{CheriyanR94}.
\end{proof}

To prove \Cref{thm:VC_undir} and \Cref{thm:VC_dir}, we will apply
our framework \cite{NanongkaiSY19} for reducing the vertex connectivity
problem to the local vertex connectivity problem. To describe the
reduction, let $T_{\text{pair}}(m,n,k,\epsilon,p)$ be the time required
for, given vertices $s$ and $t$, either finding a $st$-vertex cut
of size less than $\left\lfloor (1+\epsilon)k\right\rfloor $ or declaring
that $s$ and $t$ is $k$-connected correctly with probability at least $1-p$.
Let $T_{\text{local}}(\nu,k,\epsilon,p)$ be the time for solving
correctly probability at least $1-p$ the local vertex connectivity
problem from \Cref{cor:localVC_approx} when a volume parameter is
$\nu$, the cut-size parameter is $k$, and the accuracy parameter
is $\epsilon$.
\begin{lemma}
[\cite{NanongkaiSY19} Lemma 5.14, 5.15]\label{lem:VC_framework} There is a randomized
(Monte Carlo) algorithm that takes as input a graph $G$, a cut-size
parameter $k$, and an accuracy parameter $\epsilon>0$, and runs
in time in one of these expressions
\begin{align} 
\tilde{O}(m/\overline{\nu})\cdot(T_{\textpair}(m,n,k,\epsilon,1/\text{\ensuremath{\poly}}n)&+T_{\textlocal}(\overline{\nu},k,\epsilon,1/\text{\ensuremath{\poly}}n)) \label{eq:edge-sampling} \\ 
\ot (n/\overline{n})\cdot(T_{\textpair}(m,n,k,\epsilon,1/\text{\ensuremath{\poly}}n) &+  T_{\textlocal}(\overline{n}^2 + \overline{n}k,k,\epsilon,1/\text{\ensuremath{\poly}}n)) \label{eq:node-sampling} %
\end{align} 

where $\overline{\nu}\le m,$ and $ \overline{n}\le n$ are  optimizing parameters that can be
chosen, and either outputs a vertex cut of size less than $\left\lfloor (1+\epsilon)k\right\rfloor $
or declares that $G$ is $k$-connected w.h.p.
\end{lemma}

For completeness, we give a simple proof sketch of \Cref{eq:edge-sampling} which is used
for our algorithm for undirected graphs. The idea for other equations
is similar and also simple.
\begin{proof}[Proof sketch]
	Suppose that $G$ is not $k$-connected. It suffices to show
	an algorithm that outputs a vertex cut of size less than $\left\lfloor (1+\epsilon)k\right\rfloor $
	w.h.p. By considering both $G$ and its reverse graph (where the direction
	of each edge is reversed), there exists w.l.o.g.~a separation triple
	$(L,S,R)$ where $\vol^{\out}(L)\le\vol^{\out}(R)$. There are two
	cases.
	
	Suppose $\vol^{\out}(L)\ge\overline{\nu}$. By sampling $\tilde{O}(m/\nu)$
	pairs of edges $e=(x,x')$ and $f=(y,y')$, there exists w.h.p.~a pair
	$(e,f)$ where $x\in L$ and $y\in R$. For such pair $(x,y)$, if
	we check $x$ and $y$ is $k$-connected in time $T_{\text{pair}}(m,n,k,\epsilon,1/\text{\ensuremath{\poly}}n)$,
	we must obtain an $xy$-vertex cut of size less than $\left\lfloor (1+\epsilon)k\right\rfloor $.
	So, if we check this for each pair $(x,y)$ and we will obtain the
	cut w.h.p. 
	
	Suppose $\vol^{\out}(L)\le\overline{\nu}$. Suppose further that $\vol^{\out}(L)\in(2^{i-1},2^{i}]$. 
	By sampling $\tilde{O}(m/2^{i})$ pair of edges $e=(x,x')$,
	there exists w.h.p.~a edge $e$ where $x\in L$. For such vertex $x$,
	if we check the local vertex connectivity in time $T_{\text{local}}(2^{i},k,\epsilon,1/\text{\ensuremath{\poly}}n)$,
	then the algorithm must return a vertex cut of size less than $\left\lfloor (1+\epsilon)k\right\rfloor $.
	So, if we check this for each pair $(x,y)$ and we will obtain the
	cut w.h.p.
	
	To conclude, the running time in the first case is $\tilde{O}(m/\overline{\nu})\cdot T_{\text{pair}}(m,n,k,\epsilon,1/\text{\ensuremath{\poly}}n)$. For the second case, we try all $O(\log n)$ many $2^i$, each of which case takes  $\tilde{O}(m/2^{i})\cdot T_{\text{local}}(2^{i},k,\epsilon,1/\text{\ensuremath{\poly}}n)=\tilde{O}(m/\overline{\nu})\cdot T_{\text{local}}(\overline{\nu},k,\epsilon,1/\text{\ensuremath{\poly}}n)$
	assuming that $T_{\text{local}}(\nu,k,\epsilon,1/\text{\ensuremath{\poly}}n)\ge\nu$. This complete the proof of the running time.
	For the correctness, if $G$ is not $k$-connected, we must obtain a desired vertex cut of size $\left\lfloor (1+\epsilon)k\right\rfloor $ w.h.p. So if we do not find any cut, we declare that $G$ is $k$-connected w.h.p.
\end{proof}

\subsection{Undirected graphs}

Here, we prove \Cref{thm:VC_undir}. First, it suffices to
show an algorithm with $\tilde{O}(mk/\epsilon)$ time. Indeed, by using the
sparsification algorithm by Nagamochi and Ibaraki \cite{NagamochiI92},
we can sparsify an undirected graph in linear time so that $m=O(nk)$
and $k$-connectivity is preserved. By this preprocessing, the total
running time is $O(m)+\tilde{O}((nk)k/\epsilon))=\tilde{O}(m+nk^{2}/\epsilon)$
as desired. Next, we assume that $k\le\min\{n/4,5\delta\}$ where
$\delta$ is the minimum out-degree of $G$. If $k>5\delta$, then
it is $G$ is clearly not $k$-connected and the out-neighbor of the
vertex with minimum out-degree is a vertex cut of size less than $k$.
If $k>n/4$, then we can invoke the algorithm by Henzinger, Rao and
Gabow \cite{HenzingerRG00} for solving the problem exactly in time
$O(mn)=O(mk)$.

Now, we have $T_{\text{pair}}(m,n,k,\epsilon,p)=O(mk)$ by Ford-Fulkerson
algorithm. By repeating the algorithm from \Cref{cor:localVC_approx} $O(\log\frac{1}{p})$
times for boosting success probability, $T_{\text{local}}(\nu,k,\epsilon,p)=O(\nu k\epsilon^{-1}\log\frac{1}{p})$.
We choose $\overline{\nu}=O(\epsilon m)$ as required by \Cref{cor:localVC_approx}
and also $k\le\min\{n/4,5\delta\}$. Applying \Cref{lem:VC_framework} (\Cref{eq:edge-sampling}),
we obtain an algorithm for \Cref{thm:VC_undir} with running time
\[
\tilde{O}(m/\epsilon m)\times O(mk+(\epsilon m)k\epsilon^{-1}\log n)=\tilde{O}(mk/\epsilon).
\]

\subsection{Directed graphs}
Here, we prove \Cref{thm:VC_dir}.
We again assume that $k\le\min\{n/4,5\delta\}$ using the same reasoning as in the undirected case.
We first show how to obtain the claimed time bound for the approximate problem.
Note that the $\tilde{O}(mk/\epsilon)$-time algorithm follows by the same argument as in the undirected case, 
because both Ford-Fulkerson algorithm and the local algorithm from \Cref{cor:localVC_approx} work as well in directed graphs.

Next, we show an approximate algorithm with running time $\tilde{O}(\poly(1/\epsilon) n^{2+o(1)}\sqrt{k})$. 
We assume $k \le n^{2/3}$  (for $k \ge n^{2/3}$, we use state-of-the-art $\ot( \poly(1/\epsilon) n^{3+o(1)}/k )$-time algorithm by \cite{NanongkaiSY19}). We have $T_{\text{local}}(\nu,k,\epsilon,p)=O(\nu k\epsilon^{-1}\log\frac{1}{p})$ by \Cref{cor:localVC_approx} and
$T_{\text{pair}}(m,n,k,\epsilon,1/\poly n)= \ot(\poly(1/\epsilon) n^{2+o(1)})$ using the recent result for $(1+\eps)$-approximating the minimum $st$-vertex cut by Chuzhoy and Khanna \cite{ChuzhoyK19}. By choosing $\overline{n} = n/\sqrt{k}$ for \Cref{lem:VC_framework} (\Cref{eq:node-sampling}), we obtain an algorithm with running time
\begin{align*}
\tilde{O}(n/\overline{n})\cdot(n^{2+o(1)}\poly(1/\epsilon)+(\overline{n}^{2}k+\overline{n}k^{2})/\epsilon) 
& =\tilde{O}(\sqrt{k}\poly(1/\epsilon))\cdot(n^{2+o(1)}+n^{2}+nk^{1.5})\\
& =\tilde{O}(n^{2+o(1)}\sqrt{k}\poly(1/\epsilon)).
\end{align*}

Next, we show how to obtain the time bound for the exact problem. First, observe that we can obtain a $\ot(mk^2)$-time exact algorithm from the $\ot(mk/\epsilon)$-time approximate algorithm by setting $\epsilon < 1/k$. %
It remains to show an algorithm with the running time $\ot(k^3n + k^{3/2}m^{1/2}n)$.
By \Cref{cor:localVC_exact}, there is an exact algorithm for local vertex connectivity with running time $T_{\text{local}}(\nu,k,1/2k,p)=O(\nu k^2 \log\frac{1}{p})$. 
Also, we have $T_{\text{pair}}(m,n,k,\epsilon,p)= O(mk)$ by Ford-Fulkerson algorithm. By choosing $\overline{n}  = O(\sqrt{m/k})$ in \Cref{lem:VC_framework} (\Cref{eq:node-sampling}), we obtain an algorithm with running time 
\begin{align*}
\ensuremath{\ot(n/\overline{n})\cdot(mk+(\overline{n}^{2}+\overline{n}k)k^{2})} & =\ot(n/\overline{n})\cdot(mk+(m/k+\sqrt{mk})k^{2})\\
& =\ot(n\sqrt{k/m})\cdot(mk+k^{2.5}\sqrt{m})\\
& =\ot(k^{3/2}m^{1/2}n+k^{3}n).
\end{align*}
Note that $\overline{n}^{2}+\overline{n}k = O(m/k)$ as required by \Cref{cor:localVC_exact}.

%% file: property_testing.tex
\section{Property Testing}

In this section,  we show property testing algorithms for
distinguishing between a graph that is $k$-edge/$k$-vertex connected and a graph that is 
$\epsilon$-far from having such property with constant probability for both unbounded-degree
 and bounded-degree incident-list model. Recall that for any $\epsilon > 0$, a directed
graph $G$ is $\epsilon$-far from having a property $P$ if at least
$\epsilon m$ edge modifications are needed to make $G$ satisfy
property $P$.  We assume that $\bar{d} = m/n$ is known to the
algorithm at the beginning.%
\begin{theorem} \label{thm:main-property-testing}
For unbounded-degree model, there is a property testing algorithm for $k$-edge ($k$-vertex where
$k < n/4$) connectivity with correct probability at least $2/3$ that uses
$\ot(k^2/(\epsilon^2 \bar{d}))$ queries (same for $k$-vertex) and runs in
$\ot(k^2/(\epsilon^{11/3} \bar{d}))$ time  ($\ot(k^2/(\epsilon^{2.5}
\bar{d})) $ time for $k$-vertex).  If $\bar{d}$ is unknown, then there is
a similar algorithm that uses $\ot(k/\epsilon^2)$ queries (same for
$k$-vertex), and runs in $\ot(k/\epsilon^{8/3})$ time
($\ot(k/\epsilon^{2.5})$ time for $k$-vertex).  If $G$ is simple, then
the same algorithm for testing $k$-edge-connectivity queries at most
  $\ot(\min\{ k^2/(\bar{d}\epsilon^2), k/(\bar{d}\epsilon^3) \})$  (or $\ot(\min\{k/\epsilon^2, 1/\epsilon^3 \})$ edges if
 $\bar{d}$ is unknown), and runs in  $\ot(1/(\epsilon^{14/3} \bar{d}))$ (or
 $\ot(1/\epsilon^{11/3})$ if $\bar{d}$ is unknown). 
\end{theorem}

For bounded-degree model, we assume that $d$ is known in the beginning. 
\begin{theorem} \label{thm:main-property-testing2} 
For bounded-degree model, there is a property testing algorithm for $k$-edge ($k$-vertex where $k < n/4$) connectivity with correct probability at least $2/3$ that uses
$\ot(k/\epsilon)$ queries (same for $k$-vertex) and runs in
$\ot(k/\epsilon^{8/3})$ time  ($\ot(k/\epsilon^{1.5}
) $ time for $k$-vertex).  If $G$ is simple, then
the same algorithm for testing $k$-edge-connectivity queries at most
 $\ot(\min\{ k/\epsilon, 1/\epsilon^2 \})$. 
\end{theorem}

 We prove \Cref{thm:main-property-testing} using properties of
 $\epsilon$-far from being $k$-edge/vertex connected from
 \cite{OrensteinR11}  and \cite{FrankJ99} along with a variant of
approximate LocalEC  in \Cref{sec:test-k-edge}, and approximate LocalVC in \Cref{sec:test-k-vertex}. 

\subsection{Testing $k$-Edge-Connectivity: Unbounded-Degree Model} \label{sec:test-k-edge}
In this section, we prove \Cref{thm:main-property-testing} for testing
$k$-edge-connectivity. The key tool for our property testing
algorithm is approximate local edge connectivity in a suitable form
for the application to property testing.  We can derive the following
gap version of LocalEC in \cite{NanongkaiSY19} by essentially setting
$\epsilon = \gap/k$.  

\begin{lemma} [Implicit in \cite{NanongkaiSY19} ] \label{lem:gaplocalec}  
There is a randomized (Monte Carlo) algorithm that takes as input a
vertex $x \in V$ of an $n$-vertex $m$-edge directed graph $G = (V,E)$
represented as incidence lists, a volume parameter $\nu$, a cut-size
parameter $k \geq 1$, and ``gap'' parameter $\gap \in (0,k)$ where
$\nu < \gap\cdot m/(8k)$, queries at most $\ot(\nu k/ \gap)$ edges,
 runs in $\ot((\nu/ \gap)^{5/3}k)$ time, and 
\begin{itemize} [noitemsep]
\item if there is a vertex-set $S$ such that
 $S \ni x, \vol^{\textout}(S) \leq \nu$, and $|E(S,V-S)| < k - \gap$,
 then it returns an edge-cut of size less than $k $,
\item if there is no vertex-set $S$ such that $S \ni x,
  \vol^{\textout}(S) \leq \nu$, and $|E(S,V-S)| < k$, then it returns
  the symbol $\perp$. 
\end{itemize}
\end{lemma}

We present an algorithm for testing $k$-edge-connectivity assuming
\Cref{lem:gaplocalec}.
 
\paragraph{Algorithm.}  
\begin{enumerate}
\item Sample $\Theta (\frac{1}{\epsilon})$ vertices uniformly.
\item If any of the sampled vertex has degree less than $k$, returns the
  corresponding edge-cut. 
\item Sample $\Theta( \frac{ k \log k} { \epsilon \bar{d} })$ vertices
  uniformly  (if $\bar{d}$ is unknown, then we sample $\Theta( \frac{\log
    k}{\epsilon})$ instead).
\item For each sampled vertex $x$, and for $i \in \{0, 1 ,\ldots,
   \lfloor \log_2k  \rfloor \}$, 
 \begin{enumerate}
 \item let $\nu = 2^{i+2} \epsilon^{-1} \lfloor \log_2k \rfloor,$ and $\gap = 2^{i}-1$. 
 \item run GapLocalEC$(x, \nu , k, \gap )$ on both $G$ and $G^R$ where
   $G^R$ is $G$ with reversed edges.  
 \end{enumerate} 
\item Return an edge-cut of size less than $ k $ if any execution of GapLocalEC returns a cut. Otherwise, declare that $G$ is $k$-edge-connected.
\end{enumerate}

\paragraph{Query and Time Complexity.} We first show that the number of
edge queries is at most $\ot(k^2/(\epsilon^2 \bar{d}))$. For each
sampled vertex $x$ and $i \in \{0,1,  \ldots, \lfloor \log_2k \rfloor \}$, by
\Cref{lem:gaplocalec}, GapLocalEC queries $\ot( \nu k/\gap) =
\ot(k/\epsilon)$ edges. The result follows from we repeat $\log_2k$
times per sample, and we sample $O(k \log k /(\epsilon \bar{d}))$
times. Next, we show that the running time is $\ot(k^2/(\epsilon^{11/3}
\bar{d}))$. This follows from the same argument, but we use the
running time for GapLocalEC instead of edge-query complexity.  If
$\bar{d}$ is unknown, we can remove the term $k/\bar{d}$ from above
since we sample $\Theta( \frac{\log  k}{\epsilon})$ vertices instead.   

\paragraph{Correctness.} If $G$ is $k$-edge-connected, 
the algorithm above never returns an edge-cut. We show that if $G$ is
$\epsilon$-far from being $k$-edge-connected, then the algorithm outputs an
edge-cut of size less $ k$ with constant probability. We start with simple
observation. 
\begin{lemma}  \label{lem:mgeqnk}
If $m < nk/4$, then with constant probability, the algorithm outputs
an edge-cut of size less than $k$ at step 2. 
\end{lemma}
\begin{proof}
Suppose $m < nk/4$. There are at most $n/2$ nodes with out-degree at least
$k$. Hence, there are at least $n/2$ nodes of degree
less than $k$. In this case, we can sample $O(1)$ time where each sampled node
$x$ we check $\deg^{\textout}(x)  < k$ to get $k$-edge-cut with
constant probability.
\end{proof}
From now  we assume that 
\begin{align} \label{eq:mgeqnk4}
m \geq nk/4.
\end{align}

Next, we state important properties when $G$ is $\epsilon$-far from
being $k$-edge-connected. For any non-empty subset $X \subset V$, let
$d^{\textout}(X) = |E(X, V - X)|$, and $d^{\textin}(X) = |E(V-X,X)|$.  
\begin{theorem} [\cite{OrensteinR11} Corollary
  8]\label{thm:eps-far-edge}
A directed graph $G = (V,E)$ is $\epsilon$-far from being
$k$-edge-connected (for $k \geq 1$) if and only if there exists a
family of disjoint subsets $\{X_1, \ldots, X_t\}$ of vertices for
which either $\sum_{i}(k - d^{\textout}(X_i)) > \epsilon m$ or $
\sum_{i}(k - d^{\textin}(X_i)) > \epsilon m$. 
\end{theorem} %
Let $\mathcal{F} :=  \{X_1, \ldots, X_t\}$ as in
\Cref{thm:eps-far-edge}.  We assume without loss of generality that
\begin{align} \label{eq:eps-far-edge} 
\sum_{i}(k - d^{\textout}(X_i)) > \epsilon m.
\end{align}  Let $\mathcal{C}_{-1} = \{ X \in \mathcal{F} \colon k \leq d^{\textout}(X) \}$. For $i
\in \{0, 1,\ldots, \lfloor \log_2k \rfloor \}$, let  $\mathcal{C}_i = \{ X \in \mathcal{F}
\colon k - d^{\textout}(X) \in [2^i, 2^{i+1}) \}$. Note that 
\begin{align} \label{eq:2ileqk}
   2^i \leq k, \text{ for any } i \in \{0, \ldots, \logk \}
\end{align}and
\begin{align} \label{eq:partitionF-edge}
 \mathcal{F} = \bigsqcup_{i = -1}^{ \lfloor \log_2k \rfloor} \mathcal{C}_i
\end{align}
where $\bigsqcup$ is the disjoint union.  Let $\cC_{i, \textbig} = \{
X \in \cC_i \colon \vol^{\textout}(X) \geq  2^{i+2} \epsilon^{-1} (\logk+1) \}$, and $\cC_{i, \textsmall} = \cC_i - \cC_{i, \textbig}$. 
The following lemma is the key for the algorithm's correctness. 
\begin{lemma} \label{lemma:many-small-edge}
There is $i$ such that $ |\cC_{i,\textsmall}| \geq \epsilon n \bar{d}
/ (4k  (\logk+1)). $ If $\bar{d}$ is unknown, we have
$|\cC_{i,\textsmall}| \geq  \epsilon n/(16(\logk+1))$ instead. 
\end{lemma}
We show that \Cref{lemma:many-small-edge} implies the correctness of
the algorithm. By sampling $O(k \log k/(\epsilon \bar{d}))$ many
vertices (or $O(\log k/\epsilon)$ if $\bar{d}$ is unknown), we get an event where a sampled vertex belongs to some vertex set $X \in
\cC_{i,\textsmall}$ with constant probability (since
$\cC_{i,\textsmall}$ contains disjoint sets). We run GapLocalEC for
every $i \in \{0, 1,\ldots, \lfloor \log_2k \rfloor\}$ using $\nu =
2^{i+2}\epsilon^{-1}\logk$, and $\gap = 2^i-1$; also, there exists $i$ such that
$|\cC_{i,\textsmall}| \geq \epsilon  n \bar{d} / (4k (\logk+1))$ (or $(\epsilon n/(16(\logk+1)))$
if $\bar{d}$ is unknown)  by \Cref{lemma:many-small-edge}. Therefore,  by \Cref{lem:gaplocalec}
GapLocalEC outputs an edge-cut of size less than $k$ with constant probability. 

\begin{proof}[Proof of \Cref{lemma:many-small-edge}] 
We show that  there is $i > 0$ such that $|\cC_i| > \epsilon m /
(2^i (\logk+1) )$. First, we show
that there is $i > 0$ such that 
\begin{align}\label{eq:property:one}
\sum_{ X \in \mathcal{C}_i} (k - d^{\textout}(X)) > \epsilon m/
(  \lfloor \log_2k \rfloor+1).
\end{align}
 Suppose otherwise that for every $i$,
$\sum_{ X \in \mathcal{C}_i } (k - d^{\textout}(X)) \leq \epsilon m/
(\lfloor \log_2k \rfloor+1).$
We have $\sum_{ X \in \mathcal{F}}(k- d^{\textout}(X) )
\stackrel{(\ref{eq:partitionF-edge})} = \sum_{i=-1}^{\lfloor \log_2k
  \rfloor }  \sum_{ X \in \mathcal{C}_i} (k - d^{\textout}(X))  \leq
\sum_{i=0}^{\lfloor \log_2k \rfloor}
\sum_{ X \in \mathcal{C}_i} (k - d^{\textout}(X)) \leq \epsilon m$. However, this
contradicts  \Cref{eq:eps-far-edge} as in
\Cref{thm:eps-far-edge}. Second, we claim that for any
$i$,  $$|\mathcal{C}_i| 2^{i+1} \geq \sum_{ X \in \mathcal{C}_i} (k - d^{\textout}(X)).$$
This follows trivially from that each element $X$ in the set
$\mathcal{C}_i$, $k - d^{\textout}(X) \leq 2^{i+1}$. For $i$ that satisfies \Cref{eq:property:one} we have 
\begin{align}\label{eq:property:two}
|\mathcal{C}_i|  \geq
\sum_{ X \in \mathcal{C}_i} (k - d^{\textout}(X)) / 2^{i+1} > \epsilon
  m/ ( 2^{i+1} (\lfloor \log_2k \rfloor+1)).
\end{align}
Recall that $\cC_{i, \textbig} = \{ X \in \cC_i \colon
\vol^{\textout}(X) \geq 2^{i+2} \epsilon^{-1} (\logk+1)  \}$, and $\cC_{i,
  \textsmall} = \cC_i - \cC_{i, \textbig}$. We show that $
|\cC_{i,\textbig}| <  |\cC_i|/2$.  
Therefore, for $i$ that satisfies \Cref{eq:property:two} we have 
$$    2|C_{i,\textbig}| \leq \sum_{X \in \cC_{i, \textbig}}
\vol^{\textout}(X)/(  2^{i+1} \epsilon^{-1} (\lfloor \log_2k \rfloor+1) )
\leq \epsilon m /(2^{i+1}(\lfloor \log_2k \rfloor+1) )
\stackrel{(\ref{eq:property:two})}  < |\cC_i|. $$
The first inequality is  because the term $\vol^{\textout}(X)/( \gamma
\epsilon^{-1} (\logk+1)) \geq 2$ for each $X \in
\cC_{i,\textbig}$, we have $ \sum_{X \in \cC_{i, \textbig}}
\vol^{\textout}(X)/( 2^{i+1} \epsilon^{-1} (\lfloor \log_2k \rfloor+1) )
\geq 2|C_{i,\textbig}|. $ The second inequality is because elements in $\cC_{i, \textbig} $
are disjoint and thus $\sum_{X \in \cC_{i, \textbig}} \vol^{\textout}(X)
\leq m$. The final inequality follows from  \Cref{eq:property:two}.

For the same $i$, since  $|\cC_{i,\textbig}| <  |\cC_i|/2$, we
have \begin{align} \label{eq:last-ineq}
|\cC_{i,\textsmall}| \geq |\cC_i|/2 \geq \epsilon m / (2^{i+2}( \lfloor
       \log_2 k\rfloor+1) ) \geq   \epsilon n \bar{d} / (4k  (\logk+1)). 
\end{align}  
The last inequality follows from $m = n\bar{d}$, and \Cref{eq:2ileqk}. If $\bar{d}$ is unknown, by \Cref{eq:mgeqnk4}, the last inequality
becomes $\epsilon m / (2^{i+2} (\logk+1)) \geq  \epsilon nk / (16k
\logk) = \epsilon n / (16(\logk+1)). $ This follows from \Cref{eq:2ileqk} and \Cref{eq:mgeqnk4}. 
\end{proof}
\paragraph{An improved bound for a simple graph.} The same algorithm gives
an improved bound when $G$ is simple.   If $\epsilon \geq 4/k$, the algorithm queries at most
$\ot(k^2/(\epsilon^2 \bar{d})) = \ot(1/(\epsilon^4 \bar{d})) $ edges
(and $\ot(1/\epsilon^3)$ edges if $\bar{d}$ is unknown). Now, we assume
$\epsilon > 4/k$, we show that there are $\Omega(\epsilon n \bar{d}/
k)$ ($\Omega(\epsilon n)$ if $\bar{d}$ is unknown) many vertices with
degree less than $k $. 

\begin{lemma} \label{lem:lots-singleton}
 If $\epsilon > 4/k$, $G$ is simple, and  $\epsilon$-far from being
 $k$-edge-connected, then there exist at least $\epsilon n/2$ vertices
 ($\epsilon \bar{d}n/(8k)$ vertices if $\bar{d}$ is unknown) with
 degree less than $k$. 
\end{lemma}

\Cref{lem:lots-singleton} immediately yields the correctness of the
algorithm as number of singleton with degree less than $k$ is at least $\epsilon n/2$ vertices
 ($\epsilon \bar{d}n/(8k)$ vertices if $\bar{d}$ is unknown), and we
 sample $\Theta( k/(\epsilon \bar{d}))$ (or $\Theta(1/\epsilon)$ vertices if $\bar{d}$ is unknown) at step 1
and 2 to check if each sampled vertex has degree less than $k$.  Next, we prove \Cref{lem:lots-singleton}. 

\begin{proof} [Proof of \Cref{lem:lots-singleton}]

Let $\cC = \{ X \colon k - d^{\textout}(X)\geq 1 \}$. We claim that 
\begin{align} \label{eq:cC-g-emk} |\cC| > \epsilon m/k. \end{align} This follows from 
$$|\cC|k \geq \sum_{X \in \cC} (k - d^{\textout}(X)) \geq \sum_{X \in
  \mathcal{F}} (k - d^{\textout}(X)) > \epsilon m.$$   The first inequality follows from each term $k -
d^{\textout}(X)$ is at most $k$, and there are $|\cC|$ terms. The
second inequality follows from each $X \in \mathcal{F}
\setminus \cC$, $k - d^{\textout}(X) \leq 0$. The third inequality follows from
\Cref{eq:eps-far-edge}. 

Let $\cC_{\textbig} = \{ X \in \cC \colon \vol^{\textout}(X)\geq
2k/\epsilon \}$, and $\cC_{\textsmall} = \cC - \cC_{\textbig}$.  We
claim that \begin{align} \label{eq:cC-small-g-eps-n}|\cC_{\textsmall}| > \epsilon n/8. \end{align} First, we show that
\begin{align} \label{eq:cC-big-leq-cC} |\cC_{\textbig}| < |\cC|/2. \end{align} This follows from $$ |\cC_{\textbig}|
\leq \sum_{X \in \cC_{\textbig}} \vol^{\textout}(X)/(2k
\epsilon^{-1}) \leq  \epsilon m/(2k) < |\cC|/2.  $$
The first inequality follows from the fact that for each $X \in \cC_{\textbig}, \vol^{\textout}(X)/
(2k \epsilon^{-1}) \geq 1$. Hence,  $\sum_{X \in \cC_{\textbig}} \vol^{\textout}(X)/
(2k \epsilon^{-1}) \geq  |\cC_{\textbig}|$.  The second inequality
follows from the fact that $\cC_{\textbig}$ contains disjoint sets,
and $ \sum_{X \in \cC_{\textbig}} \vol^{\textout}(X) \leq
\vol^{\textout}(V) = m$. The last inequality follows from
\Cref{eq:cC-g-emk}. Next, we have  \begin{align} \label{eq:cctextsmall2k}|\cC_{\textsmall}|
                                     \geq |\cC|/2 \geq \epsilon m /
                                     (2k) \geq \epsilon (n\bar{d})/
                                     (2k) \geq \epsilon (n\bar{d})/2. \end{align}
The first inequality follows from \Cref{eq:cC-big-leq-cC} and that  $\cC_{\textsmall} = \cC - \cC_{\textbig}$. The second
inequality follows from \Cref{eq:cC-g-emk}. The third inequality
follows from  $m = n\bar{d}$.  If $\bar{d}$ is unknown, the last part
of \Cref{eq:cctextsmall2k} becomes $ m / (2k) \geq \epsilon (nk )/ (8k) \geq \epsilon
n/8.$ This follows from \Cref{eq:mgeqnk4}.

It suffices to show that, for each $X\in\cC_{\textsmall}$, the average
degree of vertices in $X$, which is $\frac{\vol^{\textout}(X)}{|X|}$, is less
than $k$. If this is true, then there exists node $x\in X$ where
$\deg x<k$. Since the sets in $\cC$ are disjoint, each set $X \in
\cC$ contains a vertex with degree less than $k$, and
$|\cC_{\textsmall}| > \epsilon n/8$ (by \Cref{eq:cC-small-g-eps-n}), we have that the number of
singleton vertex with degree less than $k$ is $> \epsilon n/8$, and we
are done. 

Now, fix $X\in\cC_{\textsmall}$ and we want to show that $\frac{\vol^{\textout}(X)}{|X|}<k$.
Consider three cases. If $|X|=1$, then $\frac{\vol^{\textout}(X)}{|X|}=d^{\textout}(X)<k$.
Next, if $|X|\ge2/\epsilon$, then $\frac{\vol^{\textout}(X)}{|X|}<\frac{2k/\epsilon}{2/\epsilon}=k$
as $X\in\cC_{\textsmall}$. In the last case, we have $2\le|X|<2/\epsilon\le k/2$.
Note that $\vol^{\textout}(X)\le d^{\textout}(X)+|X|^{2}$ because the graph is simple.
So, 
\[
\frac{\vol(X)}{|X|}\le\frac{d(X)+|X|^{2}}{|X|}<\frac{k}{|X|}+|X|<\frac{k}{2}+\frac{k}{2}=k.\qedhere
\]
\end{proof}

\subsection{Testing $k$-Edge-Connectivity: Bounded-Degree Model} \label{sec:test-k-edge-bounded}
In this section, we prove \Cref{thm:main-property-testing} for testing $k$-edge-connectivity for
bounded degree model.  In this model, we know the maximum out-degree
$d$. We assume that $G$ is $d$-regular, meaning that every vertex has
degree $d$. If $G$ is not $d$-regular, we can ``treat'' $G$ as if it
is $d$-regular as follows. For any list $L_v$ of size less than $ d$, and $i
\in (|L_v|, d]$, we ensure that query$(v,i)$ returns a self-loop edge (i.e.,
an edge $(v,v)$). 

\paragraph{Edge-sampling procedure.} The key property of a $d$-regular graph is that we can sample edge
uniformly as follows. We first sample a vertex $x \in V$. Then, we
make query$(x,i)$ where $i$ is an integer sampled uniformly from
$[1,d]$. 
\begin{proposition}  \label{pro:sample-edge}
For any edge $e \in E$, the probability that $e$ is sampled from the edge-sampling procedure is $1/m$.
\end{proposition}
\begin{proof}
Fix any edge $e \in E$. The edge $e$ belongs to some list
$L_v$. Therefore, the probability that $e$ is queried according to
edge-sampling procedure is
\begin{align*}
P(e \text{ is queried}) &=   P(e \text{ is queried} \mid L_v \text{ is
                          sampled}) P(L_v \text{ is sampled})  + \\ 
&   P(e \text{ is queried} \mid L_v \text{ is
                       not   sampled}) P(L_v \text{ is not sampled} ) \\
 &= P(e \text{ is queried} \mid L_v \text{ is
                          sampled}) P(L_v \text{ is sampled} ) \\
&= (1/d)(1/n) = 1/m.\qedhere
\end{align*}
\end{proof}

We present an algorithm for testing $k$-edge-connectivity for
bounded-degree model assuming \Cref{lem:gaplocalec}.
 
\paragraph{Algorithm.}  
\begin{enumerate}
\item Sample $\Theta (\frac{1}{\epsilon})$ vertices uniformly.
\item If any of the sampled vertex has degree less than $k$, returns the
  corresponding edge-cut. 
\item For each $i \in \{0,\ldots,\logk\}$, and for each $j \in
  \{0,\ldots, \logeta\}$ where $\eta_i = 2^{i+2}\epsilon^{-1}\logk$,
\begin{enumerate}
\item Sample $\Theta (\frac{ \logk\logeta }{ \epsilon 2^{j-i} }) = \tilde \Theta( \frac{1}{\epsilon 2^{j-i}}) $  edges uniformly.   
\item let $\nu = 2^{j+1},$ and $\gap = 2^{i}-1$. 
 \item run GapLocalEC$(x, \nu , k, \gap )$ on both $G$ and $G^R$ where
   $G^R$ is $G$ with reversed edges, and $x$ is a vertex from the
   sampled edge of the form $(x,y)$.  
\end{enumerate}
\item Return an edge-cut of size less than $ k $ if any execution of GapLocalEC returns a cut. Otherwise, declare that $G$ is $k$-edge-connected.
\end{enumerate}

\paragraph{Query and Time Complexity.}  We first show that the number of
edge queries is at most $\ot(k/\epsilon)$. For each
vertex $x$ from the sampled edge $(x,y)$ and
for each $(i,j)$ pair in loops, by \Cref{lem:gaplocalec}, GapLocalEC
queries $\ot( \nu k/\gap) = \ot(2^{j-i}k)$ edges, and we sample $\ot(
1/(\epsilon 2^{j-i}))$ times.  Therefore, by repeating $\ot(1)$ time, the total edge queries is at most $\ot( k/ \epsilon)$. 

Next, we show that the running time is $\ot(k/\epsilon^{8/3})$. This
follows from the same argument, but we use the running time for
GapLocalEC instead of edge-query complexity. For each iteration, the
running time is $\ot( (\nu/\gap)^{5/3}k \cdot 1/(\epsilon 2^{j-i})) = \ot(
(2^{j-i})^{2/3}k/\epsilon) = \ot( k/\epsilon^{8/3} )$. The last
inequality follows because by definition $2^j \leq 2^{i+2}
\epsilon^{-1} \logk$. 

\paragraph{Correctness.} If $G$ is $k$-edge-connected, then the algorithm
never returns any edge-cut, and we are done. Suppose $G$ is
$\epsilon$-far from being $k$-edge-connected, then we show that the
algorithm outputs an edge-cut of size less than $k$ with constant
probability.  Since $G$ is $d$-regular, we have $\bar{d} =
d$. Therefore,  we can use results from \Cref{sec:test-k-edge}. %
 Let $\mathcal{F} :=  \{X_1, \ldots, X_t\}$ as in
\Cref{thm:eps-far-edge}.  We assume without loss of generality that
\begin{align} \label{eq:eps-far-edge2} 
\sum_{i}(k - d^{\textout}(X_i)) > \epsilon m.
\end{align}  Let $\mathcal{C}_{-1} = \{ X \in \mathcal{F} \colon k \leq d^{\textout}(X) \}$. For $i
\in \{0, 1,\ldots,\logk\}$, let  $\mathcal{C}_i = \{ X \in \mathcal{F}
\colon k - d^{\textout}(X) \in [2^i, 2^{i+1}) \}$.  Let $\cC_{i,
  \textbig} = \{ X \in \cC_i \colon \vol^{\textout}(X) \geq 2^{i+2}
(\logk+1) /\epsilon \}$, and $\cC_{i, \textsmall} = \cC_i - \cC_{i,
  \textbig}$.  By \Cref{lemma:many-small-edge},  there is $i$ such
that  \begin{align} \label{eq:ccitextsmallgeqepsilon} |\cC_{i,\textsmall}| \geq \epsilon n \bar{d}
/ (4k  (\logk+1)) = \epsilon m/ (4k (\logk+1)). \end{align} This last inequality
follows since $n\bar{d} = nd = m$.  We fix $i$ as in
\Cref{eq:ccitextsmallgeqepsilon}. Let $\eta_i =
2^{i+2}\epsilon^{-1}\logk$. For $j \in \{0, 1, \ldots, \logeta \}$,
let $\cC_{i,\textsmall, j} = \{ X \in \cC_{i,\textsmall} \colon \vol^{\textout}(X) \in [2^j, 2^{j+1}) \}$. 

\begin{lemma} \label{lem:bigvol-small}
For $i$ that satisfies  \Cref{eq:ccitextsmallgeqepsilon}, there is $j$ such that $\sum_{X \in \cC_{i,\textsmall,j}}
\vol^{\textout}(X) \geq \epsilon m 2^{j-i}/ (4 (\logk+1)( \logeta+1))$. 
\end{lemma}

We show that \Cref{lem:bigvol-small} implies the correctness. By
sampling $\Theta (\frac{ \logk\logeta }{ \epsilon 2^{j-i} }) = \tilde
\Theta( \frac{1}{\epsilon 2^{j-i}}) $ edges, we get
an event where a sampled edge $(u,v)$ has $u \in X$ for some $X \in
\cC_{i,\textsmall, j}$ with constant probability (since
$\cC_{i,\textsmall,j}$ contains disjoint elements).  For each $(i,j)  \in
\{0, 1, \ldots, \logk\} \times \{0,\ldots, \logeta\}$, we run GapLocalEC with $\nu = 2^{j+1}$, and
$\gap = 2^i-1$; also, there exists $(i,j)$ such that $\sum_{X \in \cC_{i,\textsmall,j}}
\vol^{\textout}(X) \geq \epsilon  m 2^{j-i}/ (4 (\logk+1)( \logeta+1)) $ by
\Cref{lem:bigvol-small}. Therefore, by \Cref{lem:gaplocalec}, 
GapLocalEC outputs an edge-cut of size less than $k$ with constant probability. 

\begin{proof}[Proof of \Cref{lem:bigvol-small}] 
We claim that there is $j$ such that  \begin{align} \label{eq:citextsmallj} |\cC_{i,\textsmall,
                                        j}| \geq |\cC_{i,\textsmall}| /
                                        (\logeta+1)  \end{align}
Suppose otherwise. We have for all $j \in \{0, \ldots, \logeta\}, |\cC_{i,\textsmall,j}| <
|\cC_{i,\textsmall}|/ (\logeta+1) $. Therefore, $\sum_{j \in
  \{0,\ldots, \logeta \} } |\cC_{i,\textsmall,j}| <
|\cC_{i,\textsmall}|$, a contradiction.  Now, for the same $j$,  
we have 
\begin{align*} 
 \sum_{X \in \cC_{i,\textsmall,j}} \vol^{\textout}(X)& \geq
|\cC_{i,\textsmall,j}|2^j \\&\stackrel{(\ref{eq:citextsmallj})} \geq
|\cC_{i,\textsmall}|2^j/ (\logeta+1) \\&
\stackrel{(\ref{eq:ccitextsmallgeqepsilon})} \geq
\epsilon m 2^j/ (4 k (\logk+1)( \logeta+1) ) \\&
\stackrel{(\ref{eq:2ileqk})} \geq
\epsilon m 2^{j-i}/ (4 (\logk+1) (\logeta+1)) 
\end{align*} 
The first inequality is because the set $\cC_{i,\textsmall,j}$
contains disjoint elements, and that $\vol^{\textout}(X) \geq 2^j$ by
definition. %
\end{proof}

\paragraph{An improved bound for a simple graph.} The same algorithm gives
an improved bound, $\ot(\min\{ k/\epsilon, 1/\epsilon^2 \})$ queries,  when $G$
is simple.   If $\epsilon \geq 4/k$, then the algorithm queries at
most $\ot(k/\epsilon) = \ot(1/\epsilon^2) $ edges. Otherwise, $\epsilon >
4/k$, by \Cref{lem:lots-singleton},  there are $\Omega(\epsilon n)$
many vertices with degree less than $k$, and this implies that the
algorithm outputs an edge-cut of size less than $k$ at step 2.  

\subsection{Testing $k$-Vertex-Connectivity: Unbounded-Degree Model} \label{sec:test-k-vertex}

In this section, we prove \Cref{thm:main-property-testing} for testing
$k$-vertex-connectivity. The key tool for our property testing
algorithm is approximate local vertex connecitvity in a suitable form
for the application to property testing.  We can derive the following
gap version of LocalVC in \cite{NanongkaiSY19} by essentially setting
$\epsilon = \gap/k$.  

\begin{lemma} [\cite{NanongkaiSY19} Theorem 4.1] \label{lem:gaplocalvc}
There is a randomized (Monte Carlo) algorithm that takes as input a
vertex $x \in V$ of an $n$-vertex $m$-edge directed graph $G = (V,E)$
represented as incidence-lists with minimum out-degree $\delta\ge1$, a
volume parameter $\nu$, a cut-size parameter $k$, and ``gap''
parameter $\gap \in (0, k)$ where $ \nu < \gap\cdot m/ (640k)$, and $k
\leq n/4$, that queries at most $\ot( \nu k/ \gap)$ edges, runs in
$\ot( (\nu/ \gap)^{1.5}k )$ time and 
\begin{itemize}[noitemsep] 
\item  if there is a separation triple $(L,S,R)$ where $L \ni x, |S| <
  k - \gap, \vol^{\textout}(L) \leq \nu$ and either $\min_{v \in
    L}\{\deg^{\textout}(v)\} < k$ or $|S| < k -\gap$,  then it returns
  a vertex-cut  of size  less than $ k $,
\item  if there is no separation triple $(L,S,R)$ where $L \ni x, |S|
  < k, \vol^{\textout}(L) \leq \nu$, then it returns the symbol
  $\perp$.
\end{itemize}
\end{lemma} 
We present an algorithm for testing $k$-vertex-connectivity assuming
\Cref{lem:gaplocalvc}, and analysis. 

\paragraph{Algorithm.}  
\begin{enumerate}
\item Sample $\Theta(1)$ vertices uniformly. 
\item If any of the sampled vertex $x$ has out-degree less than $k$, returns
  $N(x)$. 
\item Sample $\Theta(k \log k/( \epsilon \bar{d}))$ vertices uniformly
  (if $\bar{d}$ is unknown, sample $\Theta(\log k/\epsilon)$ vertices instead). 
\item For each sampled vertex $x$, and for $i \in \{0 ,\ldots, \logk \}$, 
 \begin{enumerate}
 \item let $\nu = 2^{i+3}\logk/\epsilon,$ and $\gap = 2^{i}-1$. 
 \item run GapLocalVC$(x, \nu , k, \gap )$ on both $G$ and $G^R$ where $G^R$ is the same graph with reversed edges. 
 \end{enumerate} 
\item Return a vertex-cut of size less than $k $ if any execution of GapLocalVC returns a vertex-cut. Otherwise, declare that $G$ is $k$-vertex-connected.
\end{enumerate}

\paragraph{Query and Time Complexity.}  We first show that the number of edge queries is at most $\ot(k^2/(\epsilon^2 \bar{d}))$. For each sampled vertex $x$ and $i \in \{1, \ldots, \logk\}$, by \Cref{lem:gaplocalvc}, GapLocalVC queries $\ot( \nu k/\gap) = \ot(k/\epsilon)$ edges. The result follows from we repeat $\logk$ times per sample, and we sample $O(k \log k /(\epsilon \bar{d}))$ times. Next, we show that the running time is $\ot(k^2/(\epsilon^{2.5} \bar{d}))$. This follows from the same argument, but we use the running time for GapLocalVC instead of edge-query complexity.

\paragraph{Correctness.} If $G$ is $k$-vertex-connected, it is clear that
the GapLocalVC never returns any vertex-cut. We show that if $G$ is
$\epsilon$-far from $k$-vertex-connected, then the algorithm outputs a
vertex-cut of size less than $k$ with constant probability. We start with
simple obsevation. 

\begin{lemma}  \label{lem:mgeqnk-vertex}
If $m < nk/4$, then with constant probability, the algorithm outputs
an vertex-cut of size less than $k$ at step 2. 
\end{lemma}
\begin{proof}
Suppose $m < nk/4$. There are at most $n/2$ nodes with out-degree at least
$k$. Hence, there are at least $n/2$ nodes of degree
less than $k$. In this case, we can sample $O(1)$ time where each sampled node
$x$ we check $|N(x)| < k$ to get $k$-vertex-cut with constant probability. 
\end{proof}
From now  we assume that 
\begin{align} \label{eq:mgeqnk4-vertex2}
m \geq nk/4.
\end{align}

We start with important properties when $G$ is $\epsilon$-far from $k$-vertex-connected. We say that two separation triples $(L,S,R)$ and $(L',S',R')$ are \textit{independent} if $L \cap L' = \emptyset$ or $R \cap R' = \emptyset$.  

\begin{theorem} [\cite{OrensteinR11} Corollary 17]  \label{thm:eps-far-vertex}
If a directed graph $G = (V,E)$ is $\epsilon$-far from being
$k$-vertex-connected, then there exists a set $\mathcal{F'}$ of
pairwise independent separation triples\footnote{We use the term
  separation triple $(L,S,R)$ instead of the term \textit{one-way
    pair} $(L,R)$ used by \cite{OrensteinR11} for notational
  consistency in our paper.  These terms are
  equivalent in that there is no edge from $L$ to $R$ and our $S$ is their $V - (L \cup R)$.} such that $\sum_{(L,S,R) \in \mathcal{F'}} \max\{k-|S|,0 \} > \epsilon m$.  
\end{theorem}

Let $\mathcal{F}$ be a family of pairwise independent separation
triples of $G$ such that  \\ $ p(\mathcal{F}) := \sum_{(L,S,R) \in \mathcal{F}}( \max\{k-|S|,0 \} )$ is maximized. By \Cref{thm:eps-far-vertex}, we have $\sum_{(L,S,R) \in \mathcal{F}} \max\{k-|S|,0 \}  > \epsilon m$. 

We say that a left-partition $L$ of a separation triple $(L,S,R)$ is \textit{small} if $|L| \leq |R|$. Similarly, a right-partition $R$ is small if $|R| \leq |L|$. 
\begin{lemma}[\cite{FrankJ99} Lemma 7] \label{lem:LR-disjoint} 
The small left-partitions\footnote{In \cite{FrankJ99}, they use the
  term \textit{one-way pair} $(T,H)$, and define a tail $T$ of a pair $(T,H)$
if \textit{small} if $|T| \leq |H|$. Similarly, they define a head $H$ of a
pair $(T,H)$ to be \textit{small} if $|H| \leq |T|$. We only repharse
from ``tail'' to left-partition, and ``head'' to right-partition.} in $\mathcal{F}$ are pairwise disjoint, and the small right-partitions in $\mathcal{F}$ are pairwise disjoint.
\end{lemma} 

 Let $\mathcal{F}_L$ be the set of separation triples with small left-partitions in $\mathcal{F}$, and $\mathcal{F}_R$ be the set of separation triples with small-right partitions in $\mathcal{F}$.  
 By \Cref{thm:eps-far-vertex}, we have that $ \max \{ p(\mathcal{F}_L)
 , p(\mathcal{F}_R) \} > \epsilon m/2$. We assume without loss of generality that
 \begin{align} \label{eq:fl-large}
  p(\mathcal{F}_L) > \epsilon m/2. 
 \end{align}

Let $\mathcal{C}_{-1}= \{ (L,S,R) \in \mathcal{F}_L \colon k \leq |S| \}$. For $i \in \{0,\ldots,\logk\}$, let  $\mathcal{C}_i = \{ (L,S,R) \in \mathcal{F}_L \colon k - |S| \in [2^i, 2^{i+1}) \}$.  Let $\cC_{i, \textbig} = \{ (L,S,R) \in \cC_i \colon \vol^{\textout}(L) \geq  2^{i+3} \epsilon^{-1} (\logk+1) \}$, and $\cC_{i, \textsmall} = \cC_i - \cC_{i, \textbig}$. %
The following lemma is the key for the algorithm's correctness. 
\begin{lemma} \label{lemma:many-small}
There is $i$ such that $ |\cC_{i,\textsmall}| >  \epsilon n \bar{d}
/ (8k  (\logk+1)) $. If $\bar{d}$ is unknown, then there is $i$ such
that $|\cC_{i,\textsmall}| \geq  \epsilon n /  (32(\logk+1)) $ .%
\end{lemma}

We show that \Cref{lemma:many-small} implies the correctness of the
algorithm. By sampling $\Theta(k \log k/(\epsilon \bar{d}))$ many nodes (or
$\Theta(\log k/(\epsilon))$ vertices if $\bar{d}$ is unknown), we
get an event where $x$ belongs to some vertex set $L$ in
separation triple $(L,S,R) \in \cC_{i,\textsmall}$ with constant
probability (this follows since  $\cC_{i,\textsmall}$ contains
pairwise disjoint small left-partitions by \Cref{lem:LR-disjoint}). We run GapLocalVC for every
$i \in \{0, 1, \ldots, \logk  \}$, and  there exists $i$ such that
$|\cC_{i,\textsmall}| \geq  \epsilon n \bar{d} / (8k ( \logk+1))$ (or $|\cC_{i,\textsmall}| \geq  \epsilon n /  (32(\logk+1))$)  by
\Cref{lemma:many-small}. Therefore, by \Cref{lem:gaplocalvc},
GapLocalVC outputs a vertex-cut of size less than $k$ with constant probability.

\begin{proof}[Proof of \Cref{lemma:many-small}] 
We show that  there is $i > 0$ such that  \begin{align} \label{eq:random2} |\cC_i| > \epsilon m /
(2^{i+2} (\logk+1)). \end{align}  First, we show
that there is $i > 0$ such that 
\begin{align} \label{eq:i-satisfy-large}
\sum_{(L,S,R) \in \mathcal{C}_i} (k - |S|) > \epsilon m/ (2(\logk+1)).
\end{align}
 Suppose otherwise that for every $i$, $\sum_{(L,S,R) \in
   \mathcal{C}_i} (k - |S|) \leq \epsilon m/ (2(\logk+1))$. We have
 $\sum_{(L,S,R) \in \mathcal{F}_L}( \max\{k-|S|,0 \} )
 =\sum_{i=-1}^{\logk} \sum_{(L,S,R)\in \mathcal{C}_i} (k - |S|)
 = \sum_{i=0}^{\logk} \sum_{(L,S,R)\in \mathcal{C}_i} (k - |S|)
 \leq \epsilon m/2$. However, this contradicts 
 \Cref{eq:fl-large}. Second, we show that for any
 $i$, \begin{align} \label{eq:random-name1} |\mathcal{C}_i|2^{i+1} \geq \sum_{(L,S,R) \in \mathcal{C}_i}
 (k - |S|). \end{align} This follows trivially from that each $(L,S,R)$ in the
 set $\mathcal{C}_i$,  $k - |S| \leq 2^{i+1}$.  Therefore, for $i$ that satisfies
 \Cref{eq:i-satisfy-large}, we have 
\begin{align} \label{eq:ci-large}
|\mathcal{C}_i|  \stackrel{(\ref{eq:random-name1})} \geq
 \sum_{(L,S,R) \in \mathcal{C}_i} (k - |S|) /2^{i+1}  \stackrel{(\ref{eq:i-satisfy-large})} > \epsilon m/
 (2^{i+2} (\logk+1)). 
\end{align}

Recall that $\cC_{i, \textbig} = \{ (L,S,R) \in \cC_i \colon \vol^{\textout}(L)
\geq 2^{i+3} (\logk+1) /\epsilon \}$, and $\cC_{i, \textsmall} = \cC_i -
\cC_{i, \textbig}$. We claim that for $i$ that satisfies \Cref{eq:ci-large}, $ |\cC_{i,\textbig}| < |\cC_i|/2$. Indeed,  we have  $$  4 |C_{i,\textbig}| \leq \sum_{(L,S,R) \in \cC_{i, \textbig}} \vol^{\textout}(L)/( \gamma \epsilon^{-1} (\logk+1)) \leq m/( \gamma \epsilon^{-1} (\logk+1))  \stackrel{(\ref{eq:random2})} < 2|\cC_i|. $$
The first inequality follows because  $\vol^{\textout}(L)/( 2^{i+1} \epsilon^{-1} (\logk+1)) \geq 4$ for each $(L,S,R) \in \cC_{i,\textbig}$.The second inequality follows since left-partitions in $\cC_{i,\textbig}$ are disjoint, and $\sum_{(L,S,R) \in \cC_{i, \textbig}} \vol^{\textout}(L) \leq m$. %

Next, we have \begin{align} \label{eq:cctextsmall-vertex}
|\cC_{i,\textsmall}|
\geq |\cC_i|/2   \stackrel{(\ref{eq:random2})} > \epsilon m / (2^{i+3} (\logk+1)) \geq   \epsilon n
\bar{d} / (8k  (\logk+1)). \end{align} The first inequality follows because $|\cC_{i,\textbig}| < |\cC_i|/2$, and $ |\cC_i| =|\cC_{i,\textbig}| +|\cC_{i,\textsmall}| $. The last inequality follows because $m = n\bar{d}$, and $2^i \leq k$.  If $\bar{d}$ is unknown, the last
inequality of  \Cref{eq:cctextsmall-vertex} becomes $\epsilon m / (2^{i+3} (\logk+1))  \stackrel{(\ref{eq:mgeqnk4-vertex2})}\geq   \epsilon nk
/ (32 k  (\logk+1))=\epsilon n / (32 (\logk+1)) $. 
\end{proof}

\subsection{Testing $k$-Vertex-Connectivity: Bounded-Degree Model}

In this section, we prove \Cref{thm:main-property-testing} for testing $k$-vertex-connectivity for bounded degree model. By the same argument in \Cref{sec:test-k-edge-bounded}, we assume that $G$ is $d$-regular, and thus we can sample edge uniformly by \Cref{pro:sample-edge}.

We present an algorithm for testing $k$-edge-connectivity for
bounded-degree model assuming \Cref{lem:gaplocalvc}.
 
\paragraph{Algorithm.}  %
\begin{enumerate}
\item Sample $\Theta(1)$ vertices uniformly.
\item If any of the sampled vertex has degree less than $k$, returns its out-neighbors. 
\item For each $i \in \{0,\ldots,\logk\}$, and for each $j \in
  \{0,\ldots, \logeta\}$ where $\eta_i = 2^{i+2}\epsilon^{-1}\logk$,
\begin{enumerate}
\item Sample $\Theta (\frac{ \logk\logeta }{ \epsilon 2^{j-i} }) = \tilde \Theta( \frac{1}{\epsilon 2^{j-i}}) $  edges uniformly.   
\item let $\nu = 2^{j+1},$ and $\gap = 2^{i}-1$. 
 \item run GapLocalVC$(x, \nu , k, \gap )$ on both $G$ and $G^R$ where
   $G^R$ is $G$ with reversed edges, and $x$ is a vertex from the
   sampled edge of the form $(x,y)$.  
\end{enumerate}
\item Return a vertex-cut of size less than $ k $ if any execution of GapLocalVC returns a vertex-cut. Otherwise, declare that $G$ is $k$-vertex-connected.
\end{enumerate}

\paragraph{Query and Time Complexity.}  We first show that the number of
edge queries is at most $\ot(k/\epsilon)$. For each
vertex $x$ from the sampled edge $(x,y)$ and
for each $(i,j)$ pair in loops, by \Cref{lem:gaplocalvc}, GapLocalVC
queries $\ot( \nu k/\gap) = \ot(2^{j-i}k)$ edges, and we sample $\ot(
1/(\epsilon 2^{j-i}))$ times.  Therefore, by repeating $\ot(1)$ itereations, the total edge queries is at most $\ot( k/ \epsilon)$. 

Next, we show that the running time is $\ot(k/\epsilon^{1.5})$. This
follows from the same argument, but we use the running time for
GapLocalEC instead of edge-query complexity. For each iteration, the
running time is $\ot( (\nu/\gap)^{1.5}k \cdot 1/(\epsilon 2^{j-i})) = \ot(
(2^{j-i})^{1.5}k/(\epsilon 2^{j-i})) = \ot( k/\epsilon^{1.5} )$. The last
inequality follows because by definition $2^j \leq 2^{i+2}
\epsilon^{-1} \logk$.

\paragraph{Correctness.} If $G$ is $k$-vertex-connected, then the algorithm
never returns any vertex-cut, and we are done. Suppose $G$ is
$\epsilon$-far from being $k$-vertex-connected, then we show that the
algorithm outputs a vertex-cut of size less than $k$ with constant
probability.  Since $G$ is $d$-regular, we have $\bar{d} =
d$. Therefore,  we can use results from \Cref{sec:test-k-vertex}. Let $\mathcal{F}_L$ be the set of separation triples with small left-partitions in $\mathcal{F}$, and $\mathcal{F}_R$ be the set of separation triples with small-right partitions in $\mathcal{F}$.  
 By \Cref{thm:eps-far-vertex}, we have that $ \max \{ p(\mathcal{F}_L)
 , p(\mathcal{F}_R) \} > \epsilon m/2$. We assume without loss of generality that
 \begin{align} \label{eq:fl-large2}
  p(\mathcal{F}_L) > \epsilon m/2. 
 \end{align} 
Let $\mathcal{C}_{-1}= \{ (L,S,R) \in \mathcal{F}_L \colon k \leq |S| \}$. For $i \in \{0,\ldots,\logk\}$, let  $\mathcal{C}_i = \{ (L,S,R) \in \mathcal{F}_L \colon k - |S| \in [2^i, 2^{i+1}) \}$.  Let $\cC_{i, \textbig} = \{ (L,S,R) \in \cC_i \colon \vol^{\textout}(L) \geq  2^{i+3} \epsilon^{-1} \logk \}$, and $\cC_{i, \textsmall} = \cC_i - \cC_{i, \textbig}$. %
By \Cref{lemma:many-small}, there is $i$ such that  
\begin{align} \label{eq:csmallgepsilonnbard-vertex}
 |\cC_{i,\textsmall}| >  \epsilon n \bar{d} / (8k  (\logk+1)) = \epsilon m / (8k  (\logk+1)).
\end{align}
The last inequality follows since $n\bar{d} = nd = m$. We fix $i$ as in \Cref{eq:csmallgepsilonnbard-vertex}.  Let $\eta_i = 2^{i+3}\epsilon^{-1} \logk$. For $j \in \{0,\ldots, \logeta\},$ let $\cC_{i,\textsmall,j} = \{ (L,S,R) \in \cC_{i,\textsmall} \colon \vol^{\textout}(L) \in [2^j, 2^{j+1}) \}$. 

\begin{lemma} \label{lem:j-large-vol-vertex}
For $i$ that satisfies \Cref{eq:csmallgepsilonnbard-vertex}, there is $j$ such that $\sum_{(L,S,R) \in \cC_{i,\textsmall,j}} \vol^{\textout}(L) \geq   \epsilon m 2^{j-i}/ (8 (\logk+1)( \logeta+1))$. 
\end{lemma}

We show that \Cref{lem:j-large-vol-vertex} implies the correctness.  By
sampling $\Theta (\frac{ \logk\logeta }{ \epsilon 2^{j-i} }) = \tilde
\Theta( \frac{1}{\epsilon 2^{j-i}}) $ edges, we get
an event where a sampled edge $(u,v)$ has $u \in L$ for some $L$ from a separation triple $(L,S,R) \in \cC_{i,\textsmall, j}$ with constant probability ( since  $\cC_{i,\textsmall}$ contains
pairwise disjoint small left-partitions by \Cref{lem:LR-disjoint}).  For each $(i,j)  \in
\{0, 1, \ldots, \logk\} \times \{0,\ldots, \logeta\}$, we run GapLocalVC with $\nu = 2^{j+1}$, and
$\gap = 2^i-1$; also, there exists $(i,j)$ such that $\sum_{X \in \cC_{i,\textsmall,j}}
\vol^{\textout}(X) \geq \epsilon  m 2^{j-i}/ (8 (\logk+1) (\logeta+1)) $ by
\Cref{lem:j-large-vol-vertex}. Therefore, by \Cref{lem:gaplocalvc}, 
GapLocalVC outputs an vertex-cut of size less than $k$ with constant probability. 

\begin{proof}[Proof of \Cref{lem:j-large-vol-vertex}]
We claim that there is $j$ such that  \begin{align} \label{eq:citextsmallj-vertex} |\cC_{i,\textsmall,j}| \geq |\cC_{i,\textsmall}| / (\logeta+1) \end{align}
Suppose otherwise. We have for all $j \in \{0, \ldots, \logeta\}, |\cC_{i,\textsmall,j}| < |\cC_{i,\textsmall}|/ (\logeta +1)$. Therefore, $\sum_{j \in
  \{0,\ldots, \logeta \} } |\cC_{i,\textsmall,j}| <
|\cC_{i,\textsmall}|$, a contradiction.  Now, for the same $j$,  we have
\begin{align*} 
 \sum_{(L,S,R) \in \cC_{i,\textsmall,j}} \vol^{\textout}(L)& \geq
|\cC_{i,\textsmall,j}|2^j \\&\stackrel{(\ref{eq:citextsmallj-vertex})} \geq
|\cC_{i,\textsmall}|2^j/ (\logeta+1) \\&
\stackrel{(\ref{eq:csmallgepsilonnbard-vertex})} \geq
\epsilon m 2^j/ (8 k (\logk+1)( \logeta+1)) \\&
 \geq \epsilon m 2^{j-i}/ (8 (\logk+1)(\logeta+1)) 
\end{align*} 
The first inequality is because the small left-partitions in $\cC_{i,\textsmall,j}$ are pairwise disjoint by \Cref{lem:LR-disjoint}, and that $\vol^{\textout}(L) \geq 2^j$ by definition.  The last inequality follows since $2^i \leq k$ by definition.  
\end{proof}

%% file: local.tex
\section{An alternate algorithm for local edge connectivity}

\label{sec:localEC_alternate}

In this section, we show local algorithms for detecting an edge cut
of size $k$ and volume $\nu$ containing some seed node in time $O(\nu k^{2})$.
Both the algorithms and analysis are very simple. 
\begin{theorem}
\label{thm:localEC_approx}There is a randomized (Monte Carlo) algorithm
that takes as input a vertex $x\in V$ of an $n$-vertex $m$-edge
graph $G=(V,E)$ represented as adjacency lists, a volume parameter
$\nu$, a cut-size parameter $k\ge1$, and an accuracy parameter $\epsilon\in(0,1]$
where $\nu<\epsilon m/8$ and runs in $O(\nu k/\epsilon)$ time and
 outputs either
\begin{itemize}
\item the ``$\bot$'' symbol indicating that, with probability $1/2$,
there is no $S\ni x$ where $|E(S,V-S)|<k$ and $\vol^{\out}(S)\le\nu$,
or
\item a set $S\ni x$ where $S\neq V$, $|E(S,V-S)|<\left\lfloor (1+\epsilon)k\right\rfloor $
and $\vol^{\out}(S)\le10\nu/\epsilon$.\footnote{We note that the factor 10 in \Cref{thm:localEC_approx} can be improved.
We only use this factor for simplifying the analysis.}
\end{itemize}
\end{theorem}

By setting $\epsilon<\frac{1}{k}$, we have that $\left\lfloor (1+\epsilon)k\right\rfloor =k$.
In particular, we obtain an algorithm for the \emph{exact} problem:
\begin{corollary}
\label{thm:localEC_exact}There is a randomized (Monte Carlo) algorithm
that takes as input a vertex $x\in V$ of an $n$-vertex $m$-edge
graph $G=(V,E)$ represented as adjacency lists, a volume parameter
$\nu$, and a cut-size parameter $k\ge1$ where $\nu<m/8k$ and runs
in $O(\nu k^{2})$ time and outputs either
\begin{itemize}
\item the ``$\bot$'' symbol indicating that, with probability $1/2$,
there is no $S\ni x$ where $|E(S,V-S)|<k$ and $\vol^{\out}(S)\le\nu$,
or
\item a set $S\ni x$ where $S\neq V$, $|E(S,V-S)|<k$ and $\vol^{\out}(S)\le10\nu k$.
\end{itemize}
\end{corollary}

\begin{algorithm}
\SetKwFor{RepTimes}{repeat}{times}{end}

\RepTimes{$\left\lfloor (1+\epsilon)k\right\rfloor $}
{
	Grow a DFS tree $T$ starting from $x$ and stop once exactly $8\nu/\epsilon$ edges have been visited. 
	
	Let $E_{\DFS}$ be the set of edges visited.
	
	\lIf{$|E_{\DFS}|<8\nu/\epsilon$}{\Return $V(T)$.}
	\Else{
		Sample an edge $(y',y)\in E_{\DFS}$ uniformly.
		 
		Reverse the direction of edges in the path $P_{xy}$ in $T$ from $x$ to $y$.
	}
}
\Return{$\bot$.}

\caption{$\localEC(x,\nu,k,\epsilon)$\label{alg:local}}
\end{algorithm}

The algorithm for \Cref{thm:localEC_approx} in described in \Cref{alg:local}.
We start with the following important observation.
\begin{lemma}
\label{lem:reverse_cutsize}Let $S\subset V$ be any set where $x\in S$.
Let $P_{xy}$ be a path from $x$ to $y$. Suppose we reverse the
direction of edges in $P_{xy}$. Then, we have $|E(S,V-S)|$ and $\vol^{\out}(S)$
are both decreased exactly by one if $y\notin S$. Otherwise, $|E(S,V-S)|$
and $\vol^{\out}(S)$ stay the same.
\end{lemma}

It is clear that running time of \Cref{alg:local} is $\left\lfloor (1+\epsilon)k\right\rfloor \times O(\nu/\epsilon)=O(\nu k/\epsilon)$ 
because the DFS tree only requires $O(\nu/\epsilon)$ for visiting $O(\nu/\epsilon)$ edges. 
The two lemmas below imply the correctness of \Cref{thm:localEC_approx}
\begin{lemma}
If a set $S$ is returned, then $S\ni x$, $S\neq V$, $|E(S,V-S)|<\left\lfloor (1+\epsilon)k\right\rfloor $
and $\vol^{\out}(S)\le10\nu/\epsilon$.
\end{lemma}

\begin{proof}
If $S$ is returned, then the DFS tree $T$ get stuck at $S=V(T)$.
That is, $|E(S,V-S)|=0$ and $\vol^{\out}(S)\le8\nu/\epsilon$ at
the end of the algorithm.  Note that $x\in S$ and $S\neq V$ because
$8\nu/\epsilon<m$. Observe that the algorithm has reversed strictly
less than $\left\lfloor (1+\epsilon)k\right\rfloor $ many paths $P_{xy}$,
because the algorithm did not reverse a path in the iteration that
$S$ is returned. So \Cref{lem:reverse_cutsize} implies that, initially,
$|E(S,V-S)|<\left\lfloor (1+\epsilon)k\right\rfloor $ and, $\vol^{\out}(S)< 8\nu/\epsilon+\left\lfloor (1+\epsilon)k\right\rfloor \le10\nu/\epsilon$.
\end{proof}
\begin{lemma}
If $\bot$ is returned, then, with probability at least $1/2$, there
is no $S\ni x$ where $|E(S,V-S)|<k$ and $\vol^{\out}(S)\le\nu$.
\end{lemma}

\begin{proof}
Suppose that such $S$ exists. We will show that $\bot$ is returned
with probability less than $1/2$. Suppose that no set $S'$ is returned
before the last iteration. It suffices to show that at the beginning
of the last iteration, $|E(S,V-S)|=0$ with probability at least $1/2$.
If this is true, then the DFS tree $T$ in the last iteration will
not be able to visit more than $\nu$ edges and so will return the
set $V(T)$.

Let $k'=\left\lfloor (1+\epsilon)k\right\rfloor -1$ denote the number
of iterations excluding the last one. Let $X_{i}$ be the random variable
where $X_{i}=1$ if the sampled edge $(y',y)$ in the $i$-th iteration
of the algorithm is such that $y\in S$. Otherwise, $X_{i}=0$. As
$\vol^{\out}(S)$ never increases, observe that $E[X_{i}]\le\frac{\vol^{\out}(S)}{|E_{\DFS}|}\le\frac{\nu}{8\nu/\epsilon}=\epsilon/8$
for each $i\le k'$. Let $X=\sum_{i=1}^{k'}X_{i}$. We have $E[X]\le\epsilon k'/8$
by linearity of expectation and $\Pr[X\le\epsilon k'/4]\ge1/2$ by
Markov's inequality. So $\Pr[X\le\left\lfloor \epsilon k'/4\right\rfloor ]\ge1/2$
as $X$ is integral.

Let $Y=k'-X$. Notice that $Y$ is the number of times before the
last iteration where the algorithm samples $y\notin S$. We claim
that $k'-\left\lfloor \epsilon k'/4\right\rfloor \ge k-1$ (see the
proof at the end). Hence, with probability at least $1/2$, $Y\ge k'-\left\lfloor \epsilon k'/4\right\rfloor \ge k-1\ge|E(S,V-S)|$.
By \Cref{lem:reverse_cutsize}, if $Y\ge|E(S,V-S)|$, then $|E(S,V-S)|=0$
at the beginning of the last iteration. This concludes the proof.
\begin{claim}
$k'-\left\lfloor \epsilon k'/4\right\rfloor \ge k-1$ for $\epsilon\in[0,1]$
\end{claim}

\begin{proof}
If $\epsilon<4/k'$, then $\left\lfloor \epsilon k'/4\right\rfloor =0$,
so $k'-\left\lfloor \epsilon k'/4\right\rfloor =\left\lfloor (1+\epsilon)k\right\rfloor -1\ge k-1$.
If $\epsilon\ge4/k'$, then\footnote{The main reason we choose the factor $8$ in the number $8\nu/\epsilon$
of visited edges by the DFS is for simplifying the following inequalities.}
\begin{align*}
k'-\left\lfloor \epsilon k'/4\right\rfloor  & \ge(1-\epsilon/4)k'\\
 & \ge(1-\epsilon/4)((1+\epsilon)k-2)\\
 & \ge(1-\epsilon/4)(1+\epsilon)k-2\\
 & \ge(1+\epsilon/2)k-2\\
 & \ge k-1.
\end{align*}
where the last inequality is because $\epsilon k/2\ge\frac{4}{k'}\cdot\frac{k}{2}\ge1$
as $k'\le\left\lfloor (1+\epsilon)k\right\rfloor \le2k$.
\end{proof}
\end{proof}